\documentclass[11pt,oneside]{article}

\usepackage[letterpaper,margin=1.25in]{geometry}     
\usepackage{setspace}
\usepackage{graphicx}
\usepackage{secdot}
\usepackage{xspace}
\usepackage{booktabs}
\usepackage{url}
\usepackage{enumerate}
\usepackage{enumitem}
\usepackage{color}
\usepackage{tikz}
\usepackage{xcolor}
\usepackage{blindtext}
\usepackage[affil-it]{authblk}
\usepackage[english]{babel}
\usepackage{blkarray}
\usepackage[title]{appendix}
\usepackage{amsmath,amsthm,amsfonts,amssymb, graphicx, multicol, array}
\usepackage{algorithm} 
\usepackage{algpseudocode}
\usepackage{mathtools}
\usepackage{booktabs}
\newtheorem{lemma}{Lemma}
\newtheorem{theorem}{Theorem}

\newtheorem{definition}{Definition}
\newtheorem{problem}{Problem}

\theoremstyle{definition}

\theoremstyle{remark}

\newtheorem*{note*}{Note}

\newtheorem*{remark*}{Remark}
\theoremstyle{claimstyle}

\newtheorem*{claim*}{Claim}
\newtheorem{innercustomthm}{Theorem}
\newenvironment{customthm}[1]
  {\renewcommand\theinnercustomthm{#1}\innercustomthm}
  {\endinnercustomthm}
\newenvironment{hproof}{%
  \proof}{\endproof}

\usepackage[numbers,square,comma,sort&compress]{natbib}

\newtheorem{innercustomgeneric}{\customgenericname}
\providecommand{\customgenericname}{}

\title{Revisiting the Complexity of and Algorithms for the Graph Traversal Edit Distance and Its Variants}

\author[1]{Yutong Qiu\thanks{These authors contributed equally to this work.}}
\newcommand\CoAuthorMark{\footnotemark[\arabic{footnote}]}
\author[1]{Yihang Shen\protect\CoAuthorMark}
\author[1]{Carl Kingsford\thanks{To whom correspondence should be addressed: \texttt{carlk@cs.cmu.edu}}}

\affil[1]{Computational Biology Department, School of Computer Science, \protect\\ Carnegie Mellon University, 5000 Forbes Avenue,  Pittsburgh, PA}
\date{}

\begin{document}
\maketitle

\begin{abstract}
\normalsize
The graph traversal edit distance (GTED), introduced by Ebrahimpour Boroojeny et al.~(2018), is an elegant distance measure defined as the minimum edit distance between strings reconstructed from Eulerian trails in two edge-labeled graphs. GTED can be used to infer evolutionary relationships between species by comparing de Bruijn graphs directly without the computationally costly and error-prone process of genome assembly. Ebrahimpour Boroojeny et al.~(2018) propose two ILP formulations for GTED and claim that GTED is polynomially solvable because the linear programming relaxation of one of the ILPs always yields optimal integer solutions. The claim that GTED is polynomially solvable is contradictory to the complexity results of existing string-to-graph matching problems. 

We resolve this conflict in complexity results by proving that GTED is NP-complete and showing that the ILPs proposed by Ebrahimpour Boroojeny et al. do not solve GTED but instead solve for a lower bound of GTED and are not solvable in polynomial time. In addition, we provide the first two, correct ILP formulations of GTED and evaluate their empirical efficiency. These results provide solid algorithmic foundations for comparing genome graphs and point to the direction of heuristics. \\
The source code to reproduce experimental results is available at\\\url{https://github.com/Kingsford-Group/gtednewilp/}.

\end{abstract}




\newcommand\trails{\textrm{trails}}
\newcommand\CC{\textrm{CC}}
\newcommand\str{\textrm{str}}
\newcommand\strset{\textrm{strset}}
\newcommand\edit{\textrm{edit}}
\newcommand\PP{\textrm{P}\xspace}
\newcommand\NP{\textrm{NP}\xspace}
\newcommand\NPcomplete{\textrm{NP-complete}\xspace}
\newcommand\match{\textrm{match}\xspace}
\newcommand\GTED{\ensuremath{\textsc{GTED}}\xspace}
\newcommand\CCTED{\ensuremath{\textsc{CCTED}}\xspace}
\newcommand\FGTED{\ensuremath{\textsc{FGTED}}\xspace}
\newcommand\ETEW{\ensuremath{\textsc{ETEW}}\xspace}
\newcommand\ETRC{\ensuremath{\textsc{ETRC}}\xspace}
\newcommand\alngraph{\ensuremath{\mathcal{A}}\xspace}
\newcommand\talngraph{\ensuremath{\mathcal{AT}}\xspace}
\newcommand\falngraph{\ensuremath{\mathcal{AF}}\xspace}
\newcommand\boundarym{\ensuremath{[\partial]}\xspace}
\newcommand\Tr{\ensuremath{T}\xspace}
\newcommand\flow{\text{flow}}
\newcommand\cost{\text{cost}}
\newcommand\emedit{\text{emedit}}
\newcommand\boundarylp{the LP in~\eqref{eqn:boundarylp}-\eqref{eqn:xinit}\xspace}
\newcommand\boundaryilp{the ILP in~\eqref{eqn:boundarylp}-\eqref{eqn:xinit}\xspace}
\newcommand\Boundarylp{The LP in~\eqref{eqn:boundarylp}-\eqref{eqn:xinit}\xspace}

\newcommand\Gtedilp{The ILP in~\eqref{eqn:gtedlp_obj}-\eqref{eqn:gtedlp_last}\xspace}
\newcommand\gtedilp{the ILP in~\eqref{eqn:gtedlp_obj}-\eqref{eqn:gtedlp_last}\xspace}
\newcommand\Gtedlp{The LP in~\eqref{eqn:gtedlp_obj}-\eqref{eqn:gtedlp_last}\xspace}
\newcommand\gtedlp{the LP in~\eqref{eqn:gtedlp_obj}-\eqref{eqn:gtedlp_last}\xspace}

\section{Introduction}

 Graph traversal edit distance (GTED)~\citep{gted} is an elegant measure of the similarity between the strings represented by edge-labeled Eulerian graphs. For example, given two de Bruijn assembly graphs~\citep{pevzner2001eulerian}, computing GTED between them measures the similarity between two genomes without the computationally intensive and possibly error-prone process of assembling the genomes. Using an approximation of GTED between assembly graphs of Hepatitis B viruses,~\citet{gted} group the viruses into clusters consistent with their taxonomy. This can be extended to inferring phylogeny relationships in metagenomic communities or comparing heterogeneous disease samples such as cancer. There are several other methods to compute a similarity measure between strings encoded by two assembly graphs~\citep{polevikov2019synteny,minkin2020scalable,mangul2016reference,huntsman2015bruijn}. GTED has the advantage that it does not require prior knowledge on the type of the genome graph or the complete sequence of the input genomes. The input to the GTED problem is two unidirectional, edge-labeled Eulerian graphs, which are defined as:

\begin{definition}[Unidirectional, edge-labeled Eulerian Graph]
    A unidirectional, edge-labeled Eulerian graph is a connected directed graph $G = (V, E, \ell, \Sigma)$, with node set $V$, edge multi-set $E$, constant-size alphabet $\Sigma$, and single-character edge labels $\ell: E\rightarrow \Sigma$, such that $G$ contains an Eulerian trail that traverses every edge $e\in E$ exactly once. The unidirectional condition means that all edges between the same pair of nodes are in the same direction.
\end{definition}
Such graphs arise in genome assembly problems (e.g. the de Bruijn subgraphs). Computing GTED is the problem of computing the minimum edit distance between the two most similar strings represented by Eulerian trails each input graph.
\begin{problem}[Graph Traversal Edit Distance (\GTED)~\citep{gted}]
    \label{prob:gted}
    Given two unidirectional, edge-labeled Eulerian graphs $G_1$ and $G_2$, compute
    \begin{equation}
    \GTED(G_1,G_2) \triangleq \min_{\substack{t_1\in \trails(G_1)\\t_2\in \trails(G_2)}} \edit(\str(t_1), \str(t_2)).
    \end{equation}
    Here, $\trails(G)$ is the collection of all Eulerian trails in graph $G$, $\str(t)$ is a string constructed by concatenating labels on the Eulerian trail $t=(e_0,e_1,\dots ,e_n)$, and $\edit(s_1, s_2)$ is the edit distance between strings $s_1$ and $s_2$.
\end{problem} 

\citet{gted} claim that GTED is polynomially solvable by proposing an integer linear programming (ILP) formulation of GTED and arguing that the constraints of the ILP make it polynomially solvable. This result, however, conflicts with several complexity results on string-to-graph matching problems.~\citet{eulerw} show that it is NP-complete to determine if a string exactly matches an Eulerian tour in an edge-labeled Eulerian graph. 
Additionally,~\citet{jain2020complexity} show that it is NP-complete to compute an edit distance between a string and strings represented by a labeled graph if edit operations are allowed on the graph. On the other hand, polynomial-time algorithms exist to solve string-to-string alignment~\citep{needleman1970general} and string-to-graph alignment~\citep{jain2020complexity} when edit operations on graphs are not allowed.

We resolve the conflict among the results on complexity of graph comparisons by revisiting the complexity of and the proposed solutions to GTED\@. We prove that computing GTED is NP-complete by reducing from the {\sc Hamiltonian Path} problem, reaching an agreement with other related results on complexity. Further, we point out with a counter-example that the optimal solution of the ILP formulation proposed by~\citet{gted} does not solve GTED.

We give two ILP formulations for GTED. The first ILP has an exponential number of constraints and can be solved by subtour elimination iteratively~\citep{dantzig1954solution,dias2022minimum}. The second ILP has a polynomial number of constraints and shares a similar high-level idea of the global ordering approach~\citep{dias2022minimum} in solving the {\sc Traveling Salesman} problem~\citep{miller1960integer}. 

In~\citet{fgted}, Flow-GTED (FGTED), a variant of GTED is proposed to compare two sets of strings instead of two strings encoded by graphs. \FGTED is equal to the edit distance between the most similar sets of strings spelled by the decomposition of flows between a pair of predetermined source and sink nodes. The similarity between the sets of strings reconstructed from the flow decomposition is measured by the Earth Mover's Edit Distance~\citep{emd,fgted}. \FGTED is used to compare pan-genomes, where both the frequency and content of strings are essential to represent the population of organisms. \citet{fgted} reduce \FGTED to \GTED, and via the claimed polynomial-time algorithm of \GTED, argue that \FGTED is also polynomially solvable. We show that this claim is false by proving that \FGTED is also NP-complete.

While the optimal solution to ILP proposed in~\citet{gted} does not solve GTED, it does compute a lower bound to GTED. We characterize the cases when GTED is equal to this lower bound. In addition, we point out that solving this ILP formulation finds a minimum-cost matching between closed-trail decompositions in the input graphs, which may be used to compute the similarity between repeats in the genomes. \citet{gted} claim their proposed ILP formulation is solvable in polynomial time by arguing that the constraint matrix of the linear relaxation of the ILP is always totally unimodular. We show that this claim is false by proving that the constraint matrix is not always totally unimodular and showing that there exists optimal fractional solutions to its linear relaxation.


We evaluate the efficiency of solving ILP formulations for GTED and its lower bound on simulated genomic strings and show that it is impractical to compute GTED on larger genomes.

In summary, we revisit two important problems in genome graph comparisons: Graph Traversal Edit Distance (GTED) and its variant FGTED. We show that both GTED and FGTED are NP-complete, and provide the first correct ILP formulations for GTED. We also show that the ILP formulation proposed by~\citep{gted} is a lower bound to GTED. We evaluate the efficiency of the ILPs for GTED and its lower bound on genomic sequences. These results provide solid algorithmic foundations for continued algorithmic innovation on the task of comparing genome graphs and point to the direction of approximation heuristics.

\section{GTED and FGTED are NP-complete}

\subsection{Conflicting results on computational complexity of GTED and string-to-graph matching}
The natural decision versions of all of the computational problems described above and below are clearly in NP. Under the assumption that $\PP \neq \NP$, the results on the computational complexity of GTED and string-to-graph matching claimed in~\citet{gted} and~\citet{eulerw}, respectively, cannot be both true.

\citet{eulerw} show that the problem of determining if an input string can be spelled by concatenating edge labels in an Eulerian trail in an input graph is NP-complete. We call this problem {\sc Eulerian Trail Equaling Word}. We show in Theorem~\ref{thm:etewp} that we can reduce \ETEW to \GTED, and therefore if \GTED is polynomially solvable, then \ETEW is polynomially solvable. The complete proof is in Appendix~\ref{sec:gtedconflict}.

\begin{problem} [Eulerian Trail Equaling Word~\citep{eulerw}]
    Given a string $s\in \Sigma^*$, an edge-labaled Eulerian graph $G$, find an Eulerian trail $t$ of $G$ such that $\str(t) = s$.
\end{problem}

\begin{theorem} \label{thm:etewp}
    If $\GTED \in \PP$ then $\ETEW \in \PP$.
\end{theorem}
\begin{hproof}
    We first convert an input instance $\langle s,G\rangle$ to \ETEW into an input instance $\langle G_1, G_2 \rangle$ to \GTED by (a) creating graph $G_1$ that only contains edges that reconstruct string $s$ and (b) modifying $G$ into $G_2$ by extending the anti-parallel edges so that $G_2$ is unidirectional. We show that if $\GTED(G_1, G_2) = 0$, there must be an Eulerian trail in $G$ that spells $s$, and if $\GTED(G_1, G_2) > 0$, $G$ must not contain an Eulerian trail that spells $s$.
\end{hproof}
    
    Hence, an (assumed) polynomial-time algorithm for $\GTED$ solves $\ETEW$ in polynomial time. This contradicts Theorem~6 of~\citet{eulerw} of the NP-completeness of \ETEW (under $\PP\ne\NP$). 

\subsection{Reduction from Hamiltonian Path to GTED and FGTED}

We resolve the contradiction by showing that GTED is NP-complete. The details of the proof are in Appendix~\ref{sec:gtednpproof}.
\begin{theorem}\label{thm:gtednp}
    \GTED is NP-complete.
\end{theorem}
\begin{hproof}
    We reduce from the \textsc{Hamiltonian Path} problem, which asks whether a directed, simple graph $G$ contains a path that visits every vertex exactly once. Here simple means no self-loops or parallel edges. The reduction is almost identical to that presented in~\citet{eulerw}, and from here until noted later in the proof the argument is identical except for the technicalities introduced to force unidirectionality (and another minor change described later). 

    Let $\langle G=(V,E)\rangle$ be an instance of \textsc{Hamiltonian Path}, with $n=|V|$ vertices. We first create the Eulerian closure of $G$, which is defined as $G'=(V', E')$ where
    \begin{equation}
    V' = \{v^{in}, v^{out} : v \in V\} \cup \{w\}.
    \end{equation}
    Here, each vertex in $V$ is split into $v^{in}$ and $v^{out}$, and $w$ is a newly added vertex. 
    $E'$ is the union of the following sets of edges and their labels:
    \begin{itemize}
    \item $E_1 = \{(v^{in},  v^{out}) : v \in V\}$,  labeled \texttt{a},
    \item $E_2 = \{(u^{out}, v^{in}) : (u,v) \in E\}$,  labeled \texttt{b},
    \item $E_3 = \{(v^{out}, v^{in}) : v \in V\}$, labeled \texttt{c},
    \item $E_4 = \{(v^{in}, u^{out}) : (u,v) \in E\}$, labeled \texttt{c},
    \item $E_5 = \{(u^{in},w) : u \in V\}$, labeled \texttt{c},
    \item $E_6 = \{(w,u^{in}) : u \in V\}$, labeled \texttt{b}.
    \end{itemize}
    $G'$ is an Eulerian graph by construction but contains anti-parallel edges. We further create $G''$ from $G'$ by adding dummy nodes so that each pair of antiparallel edges is split into two parallel, length-2 paths with labels \texttt{x\#}, where \texttt{x} is the original label.

    We also create a graph $C$ that has the same number of edges as $G''$ and spells out a string
    \begin{equation}
    q = \texttt{a\#}(\texttt{b\#a\#})^{n-1}(\texttt{c\#})^{2n-1}(\texttt{c\#b\#})^{|E| + 1}.
    \end{equation}
    
    We then argue that $G$ has a Hamiltonian path if and only if $G''$ spells out the string $q$, which uses the same line of arguments and graph traversals as in~\citet{eulerw}. We then show that $\GTED(G'', C) = 0$ if and only if $G''$ spells $q$.
\end{hproof}
Following a similar argument, we show that FGTED is also NP-complete, and its proof is in Appendix~\ref{sec:fgtednp}.
\begin{theorem}\label{thm:fgtednp}
    \FGTED is NP-complete.
\end{theorem}

\section{Revisiting the correctness of the proposed ILP solutions to GTED}

In this section, we revisit two proposed ILP solutions to GTED by~\citet{gted} and show that the optimal solution to these ILP is not always equal to GTED. 

\subsection{Alignment graph}
The previously proposed ILP formulations for GTED are based on the alignment graph constructed from input graphs. The high-level concept of an alignment graph is similar to the dynamic programming matrix for the string-to-string alignment problem~\citep{needleman1970general}. 

\begin{figure}
\centering
\includegraphics[width=0.7\textwidth]{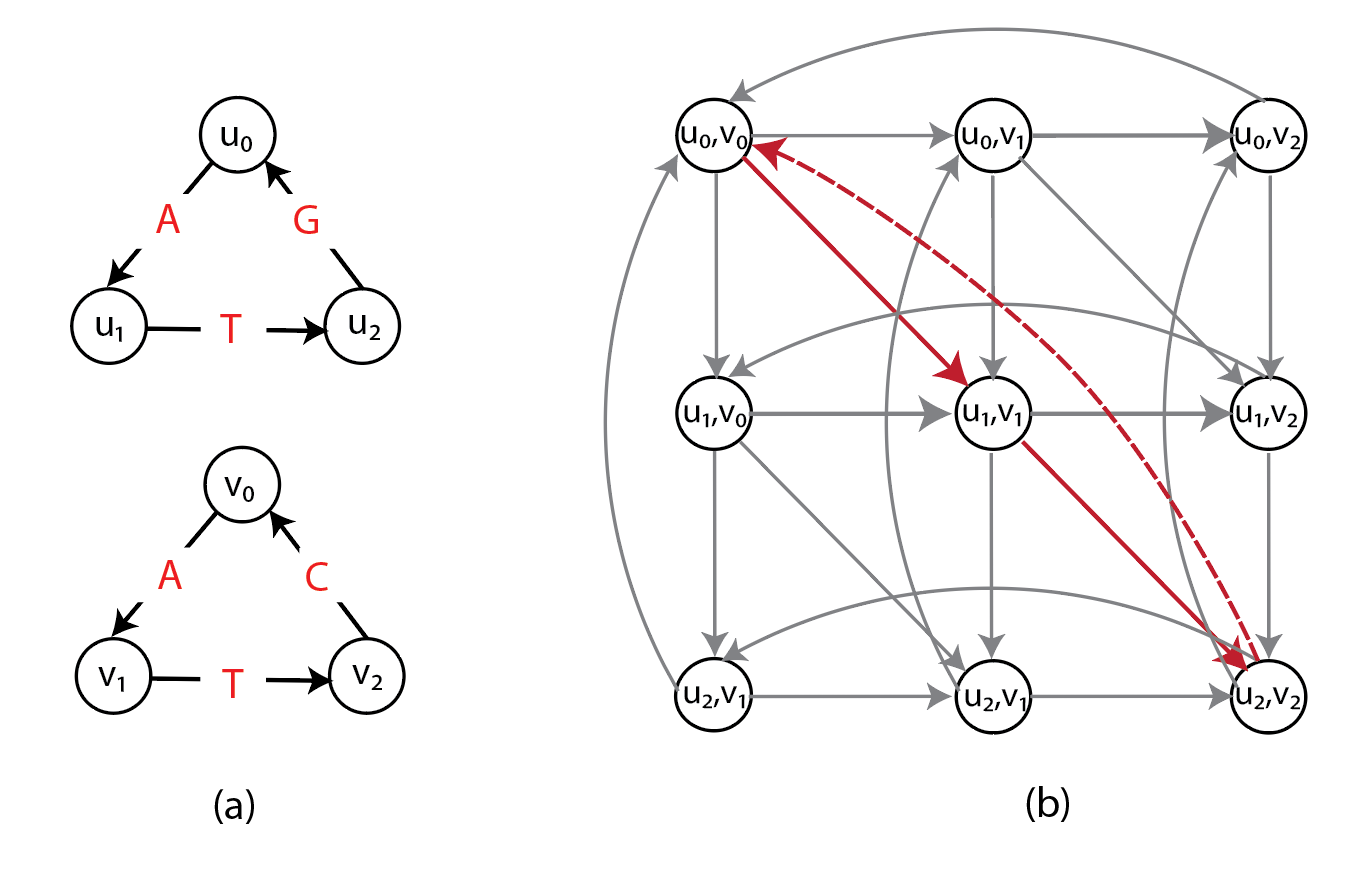}
\caption{(a) An example of two edge labeled Eulerian graphs $G_1$ (top) and $G_2$ (bottom). (b)~The alignment graph $\alngraph(G_1,G_2)$. The cycle with red edges is the path corresponding to $\GTED(G_1, G_2)$. Red solid edges are matches with cost 0 and red dashed-line edge is mismatch with cost 1.}\label{fig:alltri}
\end{figure}


\begin{definition}[Alignment graph]\label{def:align}
Let $G_1,~G_2$ be two unidirectional, edge-labeled Eulerian graphs. The \emph{alignment graph} $\alngraph(G_1,G_2) = (V, E, \delta)$ is a directed graph that has vertex set $V = V_1 \times V_2$ and edge multi-set $E$ that equals the union of the following:
\begin{description}
\item[Vertical edges]$[(u_1,u_2),(v_1,u_2)]$ for $(u_1,v_1) \in E_1$ and $u_2 \in V_2$,
\item[Horizontal edges]$[(u_1,u_2),(u_1,v_2)]$ for $u_1 \in V_1$ and $(u_2,v_2) \in E_2$,
\item[Diagonal edges] $[(u_1,u_2),(v_1,v_2)]$ for $(u_1,v_1) \in E_1$ and $(u_2,v_2) \in E_2$.
\end{description}
Each edge is associated with a cost by the cost function $\delta:E\rightarrow\mathbb{R}$.
\end{definition}

Each diagonal edge $e=[(u_1, v_1),(u_2, v_2)]$ in an alignment graph can be projected to $(u_1, v_1)$ and $(u_2, v_2)$ in $G_1$ and $G_2$, respectively. Similarly, each vertical edge can be projected to one edge in $G_1$, and each horizontal edge can be projected to one edge in $G_2$. 

We define the edge projection function $\pi_i$ that projects an edge from the alignment graph to an edge in the input graph $G_i$. We also define the path projection function $\Pi_i$ that projects a trail in the alignment graph to a trail in the input graph $G_i$. For example, let a trail in the alignment graph be $p = (e_1, e_2,\dots,e_m)$, and $\Pi_i(p) = (\pi_i(e_1), \pi_i(e_2),\dots,\pi_i(e_m))$ is a trail in $G_i$.

An example of an alignment graph is shown in Figure~\ref{fig:alltri}(b).~The horizontal edges correspond to gaps in strings represented by $G_1$, vertical edges correspond to gaps in strings represented by $G_2$, and diagonal edges correspond to the matching between edge labels from the two graphs. In the rest of this paper, we assume that the costs for horizontal and vertical edges are 1, and the costs for the diagonal edges are 1 if the diagonal edge represents a mismatch and 0 if it is a match. The cost function $\delta$ can be defined to capture the cost of matching between edge labels or inserting gaps. This definition of alignment graph is also a generalization of the alignment graph used in string-to-graph alignment~\cite{jain2020complexity}.

\subsection{The first previously proposed ILP for GTED}
\label{sec:gtedilpcounterexp}


Lemma 1 in~\citet{gted} provides a model for computing GTED by finding the minimum-cost trail in the alignment graph. We reiterate it here for completeness.

\begin{lemma}[\citep{gted}]
    For any two edge-labeled Eulerian graphs $G_1$ and $G_2$,
    \begin{equation}
        \begin{aligned}
            \GTED(G_1, G_2) = \normalfont \textrm{minimize}_c \quad& \delta(c)\\
            \normalfont \textrm{subject to}\quad& c\text{ is a trail in }\alngraph(G_1, G_2),\\
            & \Pi_i(c)\text{ is an Eulerian trail in } G_i \text{ for } i=1,2,
        \end{aligned}            
    \end{equation}        
where $\delta(c)$ is the total edge cost of $c$, and $\Pi_i(c)$ is the projection from $c$ to $G_i$.
\label{lem:gted}
\end{lemma}

An example of such a minimum-cost trail is shown in Figure~\ref{fig:alltri}(b).~\citet{gted} provide the following ILP formulation and claim that it is a direct translation of Lemma~\ref{lem:gted}:
\begin{align}
    \label{eqn:gtedlp_obj}
    \underset{x\in \mathbb{N}^{|E|}}{\textrm{minimize}}\quad &\sum_{e \in E} x_e\delta(e)\\
    \label{eqn:cycle_cst}
    \textrm{subject to}\quad & Ax = 0\\
    \label{eqn:proj_cst}
    &  \sum_{e\in E} x_eI_i(e, f) = 1\quad\text{for }i=1,2\text{ and for all }f\in E_i\\ 
    \label{eqn:gtedlp_last}
    & A_{ue} = \begin{cases}
        -1\quad \text{if $e=(u,v)\in E$ for some vertex $v\in V$}\\
        1\quad \text{if $e=(v,u)\in E$ for some $u\in V$}\\
        0 \quad \text{otherwise}
    \end{cases} 
 \end{align}
 Here, $E$ is the edge set of $\alngraph(G_1, G_2)$. $A$ is the negative incidence matrix of size $|V|\times|E|$, and $I_i(e,f)$ is an indicator function that is 1 if edge $e$ in $E$ projects to edge $f$ in the input graph $G_i$ (and 0 otherwise). We define the domain of each $x_e$ to include all non-negative integers. However, due to constraints~\eqref{eqn:proj_cst}, the values of $x_e$ are limited to either 0 or 1. We describe this ILP formulation with the assumption that both input graphs have closed Eulerian trails, which means that each node has equal numbers of incoming and outgoing edges. We discuss the cases when input graphs contain open Eulerian trails in Section~\ref{sec:newilp}.

While \gtedilp allows the solutions to select disjoint cycles in the alignment graph, the projection of edges in these disjoint cycles does not correspond to a single string represented by either of the input graphs. We show that \gtedilp does not solve GTED by giving an example where the objective value of the optimal solution to \gtedilp is not equal to GTED.


\begin{figure}
    \centering
    \includegraphics[width=0.8\textwidth]{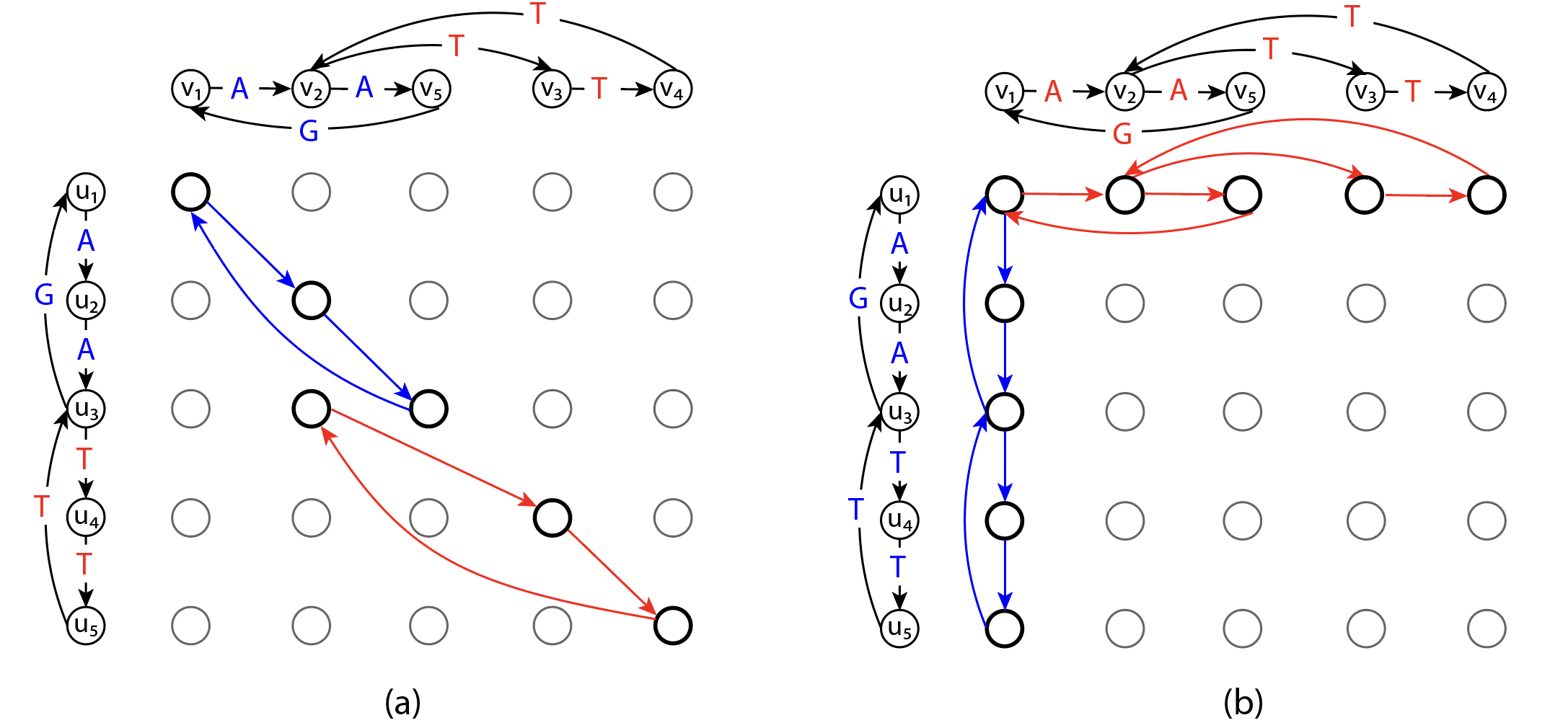}
    \caption{ (a) The subgraph in the alignment graph induced by an optimal solution to \gtedilp and \boundaryilp with input graphs on the left and top. The red and blue edges in the alignment graph are edges matching labels in red and blue font, respectively, and are part of the optimal solution to \gtedilp. The cost of the red and blue edges are zero. (b) The subgraph induced by $x^{init}$ with $s_1 = u_1$ and $s_2=v_1$ according to \boundaryilp. The rest of the edges in the alignment graph are omitted for simplicity. }
    \label{fig:separate_cycle_exp}
\end{figure}

Construct two input graphs as shown in Figure~\ref{fig:separate_cycle_exp}(a). Specifically, $G_1$ spells circular permutations of \texttt{TTTGAA}  and $G_2$ spells circular permutations of \texttt{TTTAGA}. It is clear that $\GTED(G_1, G_2) = 2$ under Levenshtein edit distance. On the other hand, as shown in Figure~\ref{fig:separate_cycle_exp}(a), an optimal solution in $\alngraph(G_1, G_2)$ contains two disjoint cycles with nonzero $x_e$ values that have a total edge cost equal to 0. This solution is a feasible solution to \gtedilp. It is also an optimal solution because the objective value is zero, which is the lower bound on \gtedilp. This optimal objective value, however, is smaller than $\GTED(G_1, G_2)$. Therefore, \gtedilp does not solve GTED since it allows the solution to be a set of disjoint components.

\subsection{The second previously proposed ILP formulation of GTED}

We describe the second proposed ILP formulation of GTED by~\citet{gted}. 
Following~\citet{gted}, we use simplices, a notion from geometry, to generalize the notion of an edge to higher dimensions. A $k$-simplex is a $k$-dimensional polytope which is the convex hull of its $k+1$ vertices. For example, a $1$-simplex is an undirected edge, and a $2$-simplex is a triangle. We use the orientation of a simplex, which is given by the ordering of the vertex set of a simplex up to an even permutation, to generalize the notion of the edge direction~\citep[p.~26]{munkres2018elements}. We use square brackets $[\cdot]$ to denote an oriented simplex. For example, $[v_0, v_1]$ denotes a 1-simplex with orientation $v_0\rightarrow v_1$, which is a directed edge from $v_0$ to $v_1$, and $[v_0, v_1, v_2]$ denotes a 2-simplex with orientation corresponding to the vertex ordering $v_0\rightarrow v_1 \rightarrow v_2\rightarrow v_0$. Each $k$-simplex has two possible unique orientations, and we use the signed coefficient to connect their forms together, e.g. $[v_0, v_1] = -[v_1, v_0]$.  

\begin{figure}
    \centering
    \includegraphics[width=0.7\columnwidth]{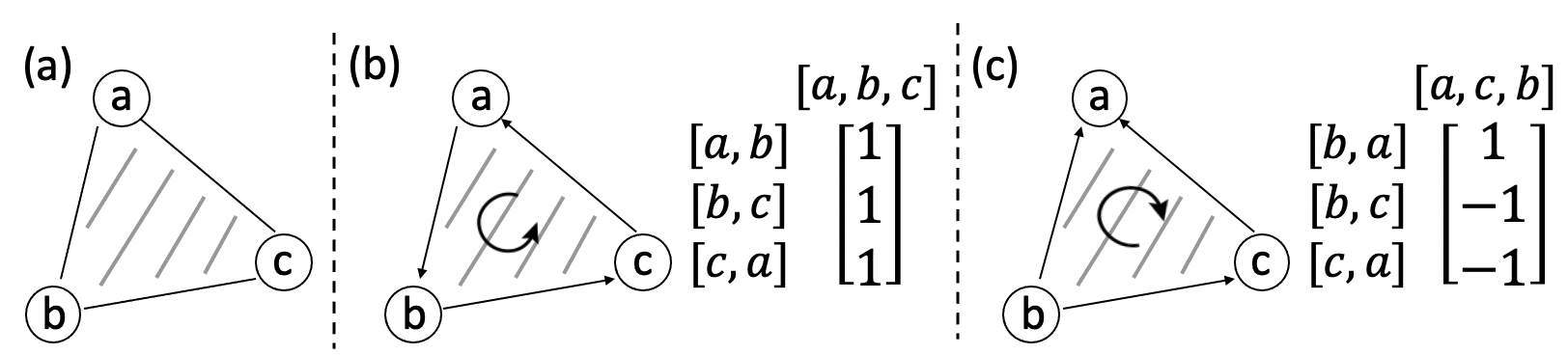}
    \caption{(a) A graph that contains an unoriented 2-simplex with three unoriented 1-simplices. (b), (c) The same graph with two different ways of orienting the simplices and the corresponding boundary matrices.}
    \label{fig:atbackground}
\end{figure}

For each pair of graphs $G_1$ and $G_2$ and their alignment graph $\alngraph (G_1,G_2)$, we define an oriented $2$-simplex set $T(G_1, G_2)$ which is the union of:
\begin{itemize}
\item $[(u_1,u_2),(v_1,u_2), (v_1,v_2)]$ for all $(u_1,v_1) \in E_1$ and $(u_2,v_2) \in E_2$, or
\item $[(u_1,u_2),(u_1,v_2),(v_1,v_2)]$ for all $(u_1,v_1) \in E_1$ and $(u_2,v_2) \in E_2$,
\end{itemize}

We use the boundary operator~\citep[p.~28]{munkres2018elements}, denoted by $\partial$, to map an oriented $k$-simplex to a sum of oriented $(k-1)$-simplices with signed coefficients.
\begin{align}
    \partial [v_0,v_1,\dots , v_k]=\sum_{i=0}^{p}(-1)^{i}[v_0,\dots , \hat{v_i},\dots , v_k],
\end{align}
where $\hat{v_i}$ denotes the vertex $v_i$ is to be deleted. Intuitively, the boundary operator maps the oriented $k$-simplex to a sum of oriented $(k-1)$-simplices such that their vertices are in the $k$-simplex and their orientations are consistent with the orientation of the $k$-simplex. For example, when $k=2$, we have:
\begin{align}
    \partial[v_0, v_1, v_2] = [v_1, v_2] - [v_0, v_2] + [v_0, v_1] 
    =[v_1, v_2] + [v_2, v_0] + [v_0, v_1].
\end{align}
We reiterate the second ILP formulation proposed in~\citet{gted}. Given an alignment graph $\alngraph(G_1, G_2)=(V, E, \delta)$ and the oriented 2-simplex set $T(G_1,G_2)$, 
\begin{equation}\label{eqn:boundarylp}
\begin{aligned}
\underset{x\in \mathbb{N}^{|E|},y\in \mathbb{Z}^{|T(G_1,G_2)|}}{\textrm{minimize}}\quad & \sum_{e \in E} x_e\delta(e)\\
\textrm{subject to}\quad & x = x^{init} + [\partial] y
\end{aligned}
\end{equation}
Entries in $x$ and $y$ correspond to 1-simplices and 2-simplices in $E$ and $T(G_1, G_2)$, respectively. $[\partial]$ is a $|E|\times |T(G_1, G_2)|$ boundary matrix where each entry $[\partial]_{i,j}$ is the signed coefficient of the oriented 1-simplex (the directed edge) in $E$ corresponding to $x_i$ in the boundary of the oriented 2-simplex in $T(G_1, G_2)$ corresponding to $y_j$. The index $i,j$ for each 1-simplex or 2-simplex is assigned based on an arbitrary ordering of the 1-simplices in $E$ or the 2-simplices in $T(G_1,G_2)$. An example of the boundary matrix is shown in Figure~\ref{fig:atbackground}. $\delta(e)$ is the cost of each edge. $x^{init} \in \mathbb{R}^{|E|}$ is a vector where each entry corresponds to a 1-simplex in $E$ with $|E_1|+|E_2|$ nonzero entries that represent one Eulerian trail in each input graph. $x^{init}$ is a feasible solution to the ILP. Let $s_1$ be the source of the Eulerian trail in $G_1$, and $s_2$ be the sink of the Eulerian trail in $G_2$. Each entry in $x^{init}$ is defined by
\begin{align}\label{eqn:xinit}
    x_e^{init} = \begin{cases}
    1\quad &\text{if}~e=[(u_1, s_2),(v_1, s_2)]~\text{or}~e=[(s_1, u_2),(s_1, v_2)],\\
    0\quad &\text{otherwise}.
    \end{cases}
\end{align}
 If the Eulerian trail is closed in $G_i$, $s_i$ can be any vertex in $V_i$. An example of $x^{init}$ is shown in Figure~\ref{fig:separate_cycle_exp}(b).

We provide a complete proof in Section~\ref{sec:ilpequivalence} of the Appendix that \gtedilp is equivalent to \boundaryilp. Therefore, the example we provided in Section~\ref{sec:gtedilpcounterexp} is also an optimal solution to \boundaryilp but not a solution to GTED. Thus, \boundaryilp does not always solve GTED.

\section{New ILP solutions to GTED}
\label{sec:newilp}

\begin{figure}
    \centering
    \includegraphics[width=0.8\textwidth]{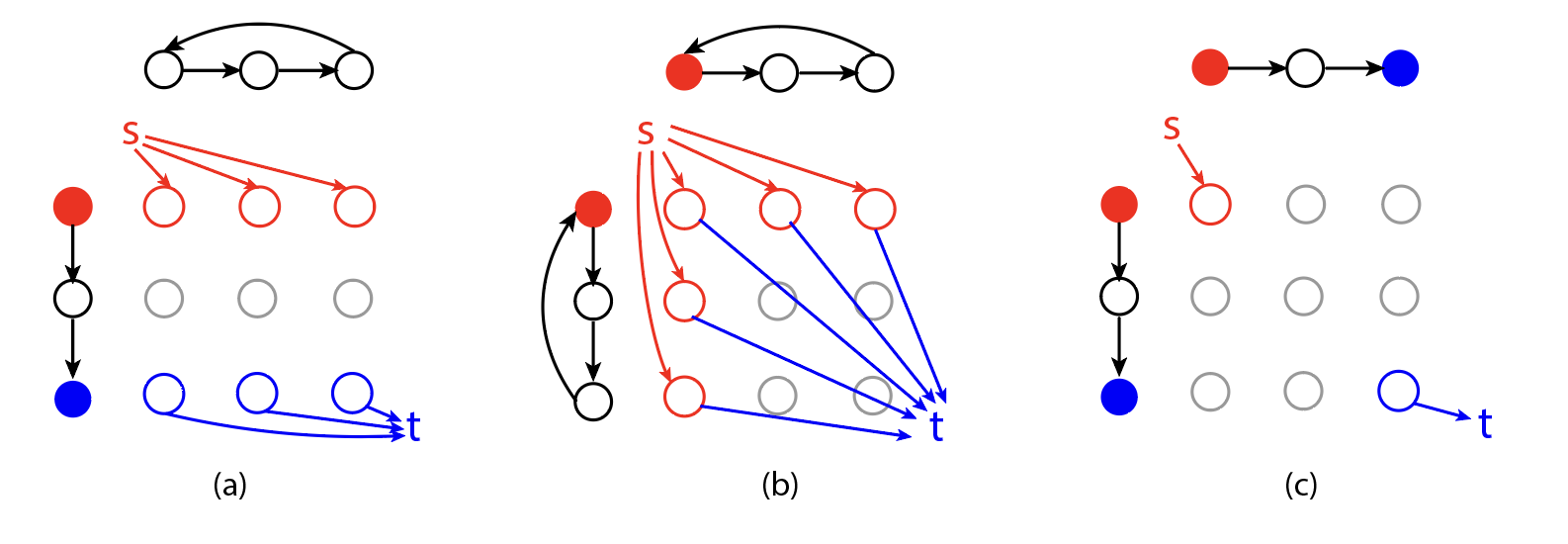}
    \caption{Modified alignment graphs based on input types. (a) $G_1$ has open Eulerian trails while $G_2$ has closed Eulerian trails. (b) Both $G_1$ and $G_2$ have closed Eulerian trails. (c) Both $G_1$ and $G_2$ have open Eulerian trails. Solid red and blue nodes are the source and sink  nodes of the graphs with open Eulerian trails. ``s'' and ``t'' are the added source and sink nodes. Colored edges are added alignment edges directing from and to source and sink nodes, respectively.}
    \label{fig:input_exp_eulerian}
\end{figure} 

To ensure that our new ILP formulations are applicable to input graphs regardless of whether they contain an open or closed Eulerian trail, we add a source node $s$ and a sink node $t$ to the alignment graph. Figure~\ref{fig:input_exp_eulerian} illustrates three possible cases of input graphs.

\begin{enumerate}
\item If only one of the input graphs has closed Eulerian trails, wlog, let $G_1$ be the input graph with open Eulerian trails. Let $a_1$ and $b_1$ be the start and end of the Eulerian trail that have odd degrees. Add edges $[s, (a_1, v_2)]$ and $[(b_1, v_2), t]$ to $E$ for all nodes $v_2\in V_2$ (Figure~\ref{fig:input_exp_eulerian}(a)).

\item If both input graphs have closed Eulerian trails, let $a_1$ and $a_2$ be two arbitrary nodes in $G_1$ and $G_2$, respectively. Add edges $[s, (a_1, v_2)]$, $[s, (v_1, a_2)]$, $[(a_1, v_2), t]$ and $[(v_1, a_2), t]$ for all nodes $v_1\in V_1$ and $v_2\in V_2$ to $E$ (Figure~\ref{fig:input_exp_eulerian}(b)).

\item If both input graphs have open Eulerian trails, add edges $[s, (a_1, a_2)]$ and $[t,(b_1, b_2)]$, where $a_i$ and $b_i$ are start and end nodes of the Eulerian trails in $G_i$, respectively (Figure~\ref{fig:input_exp_eulerian}(c)).
\end{enumerate}

According to Lemma~\ref{lem:gted}, we can solve $\GTED(G_1, G_2)$ by finding a trail in $\alngraph(G_1, G_2)$ that satisfies the projection requirements. This is equivalent to finding a $s$-$t$ trail in $\alngraph(G_1, G_2)$ that satisfies constraints:
\begin{align}
    \sum_{(u,v)\in E} x_{uv} I_i((u,v), f) = 1\quad\text{for all }(u,v)\in E,f\in G_i,~u\neq s,~v\neq t,\label{cons:projection}
\end{align}
where $I_i(e,f) = 1$ if the alignment edge $e$ projects to $f$ in $G_i$.  An optimal solution to GTED in the alignment graph must start and end with the source and sink node because they are connected to all possible starts and ends of Eulerian trails in the input graphs.

Since a trail in $\alngraph(G_1, G_2)$ is a flow network, we use the following flow constraints to enforce the equality between the number of in- and out-edges for each node in the alignment graph except the source and sink nodes.
\begin{align}
    \sum_{(s,u)\in E}x_{su} &= 1\label{cons:source_edges}\\
    \sum_{(v,t)\in E}x_{vt} &= 1\label{cons:sink_edges}\\
    \sum_{(u,v)\in E} x_{uv} &= \sum_{(v,w)\in E} x_{vw} \quad\text{for all }v\in V\label{cons:balance}
\end{align}
Constraints~\eqref{cons:projection} and~\eqref{cons:balance} are equivalent to constraints~\eqref{eqn:proj_cst} and~\eqref{eqn:cycle_cst}, respectively. Therefore, we rewrite \gtedilp in terms of the modified alignment graph. 
\begin{equation}
    \label{eqn:lowboundilp}
    \tag{lower bound ILP}
    \begin{aligned}
        \underset{x\in \mathbb{N}^{|E|}}{\textrm{minimize}}\quad &\sum_{e \in E} x_e\delta(e)\\
        \text{subject to}\quad &\text{constraints~\eqref{cons:projection}--\eqref{cons:balance}}.
    \end{aligned}
\end{equation}

As we show in Section~\ref{sec:gtedilpcounterexp}, constraints~\eqref{cons:projection}-\eqref{cons:balance} do not guarantee that the ILP solution is one trail in $\alngraph(G_1, G_2)$, thus allowing several disjoint covering trails to be selected in the solution and fails to model GTED correctly. We show in Section~\ref{sec:ccted} that the solutions to this ILP is a lower bound to GTED.

According to Lemma 1 in \citet{dias2022minimum}, a subgraph of a directed graph $G$ with source node $s$ and sink node $t$ is a $s$-$t$ trail if and only if it is a flow network and every strongly connected component (SCC) of the subgraph has at least one edge outgoing from it. Thus, in order to formulate an ILP for the GTED problem, it is necessary to devise constraints that prevent disjoint SCCs from being selected in the alignment graph. In the following, we describe two approaches for achieving this.

\subsection{Enforcing one trail in the alignment graph via  constraint generation}
\label{sec:iterativeilp}


Section 3.2 of~\citet{dias2022minimum} proposes a method to design linear constraints for eliminating disjoint SCCs, which can be directly adapted to our problem. Let $\mathcal{C}$ be the collection of all strongly connected subgraphs of the alignment graph $\alngraph (G_1,G_2)$. We use the following constraint to enforce that the selected edges form one $s$-$t$ trail in the alignment graph:
\begin{align}
 &\text{If }\sum_{(u,v)\in E(C)}x_{uv}=|E(C)|\text{, then }\sum_{(u,v)\in \varepsilon^{+}(C)}x_{uv}\geq 1 \quad\text{for all }C\in \mathcal{C},\label{cons:scc}
 \end{align}
where $E(C)$ is the set of edges in the strongly connected subgraph $C$ and $\varepsilon^{+}(C)$ is the set of edges $(u,v)$ such that $u$ belongs to $C$ and $v$ does not belong to $C$. $\sum_{(u,v)\in E(C)}x_{uv}=|E(C)|$ indicates that $C$ is in the subgraph of $\alngraph (G_1,G_2)$ constructed by all edges $(u,v)$ with positive $x_{uv}$, and $\sum_{(u,v)\in \varepsilon^{+}(C)}x_{uv}\geq 1$ guarantees that there exists an out-going edge of $C$ that is in the subgraph. 

We use the same technique as~\citet{dias2022minimum} to linearize the ``if-then'' condition in~\eqref{cons:scc} by introducing a new variable $\beta$ for each strongly connected component:
 \begin{align}
    &\sum_{(u,v)\in E(C)}x_{uv}\geq |E(C)|\beta_{C}\quad\text{for all }C\in \mathcal{C} 
    \label{cons:scclinear1}\\
    &\sum_{(u,v)\in E(C)}x_{uv}-|E(C)|+1-|E(C)|\beta_{C}\leq 0\quad\text{for all }C\in \mathcal{C} \label{cons:scclinear2}\\
    & \sum_{(u,v)\in \varepsilon^{+}(C)}x_{uv}\geq \beta_{C}\quad\text{for all }C\in \mathcal{C}\label{cons:scclinear3}\\
    & \beta_{C}\in\{0,1\}\quad\text{for all }C\in \mathcal{C} \label{cons:scclinear4}
\end{align}

To summarize, given any pair of unidirectional, edge-labeled Eulerian graphs $G_1$ and $G_2$ and their alignment graph $\alngraph(G_1,G_2)=(V,E,\delta)$, $\GTED(G_1, G_2)$ is equal to the optimal solution of the following ILP formulation:
\begin{equation}
\tag{exponential ILP}
 \begin{aligned}
  \label{eqn:newgtedlp_obj}
     \underset{x\in \{0,1\}^{|E|}}{\textrm{minimize}}\quad &\sum_{e \in E} x_e\delta(e)\\
     \textrm{subject to}\quad &\text{constraints~\eqref{cons:projection}--\eqref{cons:balance} and } \\ &\text{constraints~\eqref{cons:scclinear1}--\eqref{cons:scclinear4}.}
 \end{aligned}
 \end{equation}
This ILP has an exponential number of constraints as there is a set of constraints for every strongly connected subgraph in the alignment graph. To solve this ILP more efficiently, we can use the procedure similar to the iterative constraint generation procedure in~\citet{dias2022minimum}. Initially, solve the ILP with only constraints~\eqref{cons:projection}-\eqref{cons:balance}. Create a subgraph, $G'$, induced by edges with positive $x_{uv}$. For each disjoint SCC in $G'$ that does not contain the sink node, add constraints~\eqref{cons:scclinear1}-\eqref{cons:scclinear4} for edges in the SCC and solve the new ILP. Iterate until no disjoint SCCs are found in the solution. 

\begin{algorithm*}[!htbp]
    \caption{Iterative constraint generation algorithm to solve~\eqref{eqn:newgtedlp_obj}}
    \begin{algorithmic}[1]
        \State \textbf{Input} Two unidirectional, edge-labeled Eulerian graphs and their alignment graph
        \State $\mathcal{C}\gets\emptyset$
        \While {true}
            \State Solve the ILP~\eqref{eqn:newgtedlp_obj} with $\mathcal{C}$
            \If {the ILP variables $x_{uv}$ induce a strongly connected component $C$ not satisfying~\eqref{cons:scc}}
                \State $\mathcal{C}=\mathcal{C}\cup \{C\}$
            \Else
            \State\Return the optimal ILP value and the corresponding optimal solution $x$
            \EndIf
        \EndWhile
    \end{algorithmic}
    \label{alg:iter_gtedlp}
\end{algorithm*}

\subsection{A compact ILP for GTED with polynomial number of constraints}
\label{sec:compactilp}
In the worst cases, the number of iterations to solve~\eqref{eqn:newgtedlp_obj} via constraint generation is exponential. As an alternative, we introduce a compact ILP with only a polynomial number of constraints. The intuition behind this ILP is that we can impose a partially increasing ordering on all the edges so that the selected edges forms a $s$-$t$ trail in the alignment graph. This idea is similar to the Miller-Tucker-Zemlin ILP formulation of the \textsc{Travelling Salesman} problem (TSP)~\citep{miller1960integer}.


We add variables $d_{uv}$ that are constrained to provide a partial ordering of the edges in the $s$-$t$ trail and set the variables $d_{uv}$ to zero for edges that are not selected in the $s$-$t$ trail. Intuitively, there must exist an ordering of edges in a $s$-$t$ trail such that for each pair of consecutive edges $(u,v)$ and $(v,w)$, the difference in their order variable $d_{uv}$ and $d_{vw}$ is 1. Therefore, for each node $v$ that is not the source or the sink, if we sum up the order variables for the incoming edges and outgoing edges respectively, the difference between the two sums is equal to the number of selected incoming/outgoing edges. Lastly, the order variable for the edge starting at source is 1, and the order variable for the edge ending at sink is the number of selected edges. This gives the ordering constraints as follows:
\begin{align}
    \text{If } x_{uv} = 0,\text{ then }d_{uv} &= 0 \quad\text{for all }(u,v)\in E \label{cons:zero_order}\\
    \sum_{(v,w)\in E}d_{vw} - \sum_{(u,v)\in E} d_{uv} &= \sum_{(v,w)\in E} x_{vw}\quad\text{for all }v\in V\setminus\{s, t\} \label{cons:diff_order}\\
    \sum_{(s,u)\in E}d_{su} &= 1\label{cons:source_order}\\
    \sum_{(v,t)\in E}d_{vt} &= \sum_{(u,v)\in E} x_{uv}\label{cons:sink_order}
\end{align}
We enforce that all variables $x_e\in\{0,1\}$ and $d_e\in \mathbb{N}$ for all $e\in E$.

The ``if-then'' statement in Equation~\eqref{cons:zero_order} can be linearized by introducing an additional binary variable $y_{uv}$ for each edge~\citep{bradley1977applied,dias2022minimum}:
\begin{align}
    -x_{uv} - |E|y_{uv} &\leq -1\label{cons:orderlinear1}\\
    d_{uv} - |E|(1-y_{uv}) &\leq 0 \label{cons:orderlinearmid}\\
    y_{uv} &\in \{0,1\}\label{cons:orderlinear2}.
\end{align}
Here, $y_{uv}$ is an indicator of whether $x_{uv} \geq 0$. The coefficient $|E|$ is the number of edges in the alignment graph and also an upper bound on the ordering variables. When $y_{uv} = 1$, $d_{uv}\leq 0$, and $y_{uv}$ does not impose constraints on $x_{uv}$. When $y_{uv} = 0$, $x_{uv}\geq 1$, and $y_{uv}$ does not impose constraints on $d_{uv}$.

\subsection{Correctness of~\eqref{eqn:compactgtedilp_obj} for GTED}
To show that the optimal objective value of~\eqref{eqn:compactgtedilp_obj} is equal to GTED, we show that the optimal solutions to~\eqref{eqn:compactgtedilp_obj} always form one connected component.
\label{sec:compactilpproof}
\begin{lemma} 
    Let $x_e$ and $d_e$ be ILP variables. Let $G'$ be a subgraph of $\alngraph(G_1, G_2)$ that is induced by edges with $x_e = 1$.  If $x_e$ and $d_e$ satisfy constraints~\eqref{cons:projection}-\eqref{cons:sink_order} for all $e\in E$, $G'$ is connected with one trail from $s$ to $t$ that traverses each edge in $G'$ exactly once.    \label{lem:compactilp}
\end{lemma}

\begin{proof}
We prove the lemma in 2 parts: (1) all nodes except $s$ and $t$ in $G'$ have an equal number of in- and out-edges, (2) $G'$ contains only one connected component.

The first statement holds because the edges of $G'$ form a flow from $s$ to $t$, and is enforced by constraints~\eqref{cons:balance}.


We then show that $G'$ does not contain isolated subgraphs that are not reachable from $s$ or $t$. Due to constraint~\eqref{cons:balance}, the only possible scenario is that the isolated subgraph is strongly connected. Suppose for contradiction that there is a strongly connected component, $C$, in $G'$ that is not reachable from $s$ or $t$. 
    
The sum of the left hand side of constraint~\eqref{cons:diff_order} over all vertices in $C$ is
    \begin{align}
        \sum_{v\in C} \Big(\sum_{(u,v)\in C} d_{uv} - \sum_{(v,w)\in C} d_{vw}\Big)
        & = \sum_{v\in C} \sum_{(u,v)\in C} d_{uv} - \sum_{v\in C} \sum_{(v,w)\in C} d_{vw} \\
        &= \sum_{(u,v)\in E(C)} d_{uv} - \sum_{(v,w)\in E(C)} d_{vw} = 0.
    \end{align}
However, the right-hand side of the same constraints is always positive. Hence we have a contradiction. Therefore, $G'$ has only one connected component.
\end{proof}

Due to Lemma~\ref{lem:gted} and Lemma~\ref{lem:compactilp}, given input graphs $G_1$ and $G_2$ and the alignment graph $\alngraph(G_1, G_2)$, $\GTED(G_1, G_2)$ is equal to the optimal objective of
\begin{equation}
\tag{compact ILP}
    \begin{aligned}
        \label{eqn:compactgtedilp_obj}
     \underset{x\in \{0,1\}^{|E|}}{\textrm{minimize}}\quad  &\sum_{e \in E} x_e\delta(e)\\
     \textrm{subject to}\quad &\text{constraints~\eqref{cons:projection}--\eqref{cons:balance},} \\ 
     & \text{constraints~\eqref{cons:diff_order}--\eqref{cons:sink_order}}\\  
     & \text{and constraints~\eqref{cons:orderlinear1}--\eqref{cons:orderlinear2}.} 
    \end{aligned}
\end{equation}

\section{Closed-trail Cover Traversal Edit Distance}
\label{sec:ccted}
While the~\eqref{eqn:lowboundilp} and \boundaryilp do not solve GTED, the optimal solution to these ILPs is a lower bound of GTED. These ILP formulations also solve an interesting variant of GTED, which is a local similarity measure between two genome graphs. We call this variant Closed-trail Cover Traversal Edit Distance (CCTED). In the following, we provide the formal definition of the CCTED problem and then show that the~\eqref{eqn:lowboundilp} is the correct ILP formulation for solving CCTED. 

We first introduce the min-cost item matching problem between two multi-sets. Let two multi-sets of items be $S_1$ and $S_2$, and, wlog, let $|S_1| \leq |S_2|$. Let $c: (S_1 \cup \{\epsilon\}) \times S_2 \to \mathbb{N}$ be the cost of matching either an empty item $\epsilon$ or an item in $S_1$ with an item in $S_2$. Given $S_1$, $S_2$ and the cost function $c$, min-cost matching problem finds a matching, $\mathcal{M}_c(S_1, S_2)$, such that each item in $S_1\cup\{\epsilon\}^{|S_2|-|S_1|}$ is matched with exactly one distinct item in $S_2$ and the total cost of the matching, $\sum_{(s_1, s_2)\in \mathcal{M}_c(S_1, S_2)} c(s_1, s_2)$, is minimized.

The min-cost item matching problem is similar to the Earth Mover's Distance defined in~\cite{pele2008linear}, except that only integral units of items can be matched and the cost of matching an empty item with another item is not constant. Similar to the Earth Mover's Distance, the min-cost item matching problem can be computed using the ILP formulation of the min-cost max-flow problem~\cite{emd,fgted}. When the cost is the edit distance, the cost to match $\epsilon$ with a string is equal to the length of the string.

Define traversal edit distance, $edit_t(t_1, t_2)$ as the edit distance between the strings constructed from a pair of trails $t_1$ and $t_2$. In other words, $edit_t(t_1, t_2) = edit(\str(t_1), \str(t_2))$. \CCTED is defined as:

\begin{problem} [Closed-Trail Cover Traversal Edit Distance (\CCTED)]
Given two unidirectional, edge-labeled Eulerian graphs $G_1$ and $G_2$ with closed Eulerian trails, compute

\begin{equation}
\CCTED(G_1, G_2) \triangleq \min_{\substack{C_1\in \CC(G_1),\\C_2\in \CC(G_2)}} \sum_{\substack{(t_1, t_2)\in\mathcal{M}_{\edit_t}(C_1, C_2)}}\edit(\str(t_1), \str(t_2)),
\end{equation}
Here, $\CC(G)$ denotes the collection of all possible sets of edge-disjoint, closed trails in $G$, such that every edge in $G$ belongs to exactly one of these trails. Each element of $\CC(G)$ can be interpreted as a cover of $G$ using such trails. $\mathcal{M}_{\edit_t}(C_1, C_2)$ is a min-cost matching between two covers using the traversal edit distance as the cost.
\label{prob:ccted}
\end{problem}

\CCTED is likely a more suitable metric comparison between genomes that undergo large-scale rearrangements. This analogy is to the relationship between the synteny block comparison~\cite{polevikov2019synteny} and the string edit distance computation, where the former is more often used in interspecies comparisons and in detecting segmental duplications~\cite{bourque2004reconstructing,vollger2022segmental} and the latter is more often seen in intraspecies comparisons. 

Following similar ideas as Lemma~\ref{lem:gted}, we can compute CCTED by finding a set of closed trails in the alignment graph such that the total cost of alignment edges is minimized, and the projection of all edges in the collection of selected trails is equal to the multi-set of input graph edges. 

\label{sec:gtedcctedproofs}
\begin{lemma}
    \label{lem:ccted}
    For any two edge-labeled Eulerian graphs $G_1$ and $G_2$, 
    \begin{align}
        \CCTED(G_1, G_2) = \underset{C}{\textrm{minimize}}\quad &\sum_{c\in C} \delta(c) \label{eqn:cctedlemobj}\\
        \text{subject to}\quad &C \text{ is a set of closed trails in }\alngraph(G_1,G_2),\nonumber\\
        &\bigcup_{e\in C}\Pi_i(e) = E_i\quad \text{for }i=1,2,\label{eqn:cctedlemproj}
    \end{align}
where $C$ is a collection of trails and $\delta(c)$ is the total cost of edges in trail $c$.

\end{lemma}
\begin{proof}
Given any pair of covers $C_1\in\CC(G_1)$ and $C_2\in\CC(G_2)$ and their min-cost matching based on the edit distance $\mathcal{M}_{\edit_t}(C_1, C_2)$, we can project each pair of matched closed trailed to a closed trail in the alignment graph. For a matching between a trail and the empty item $\epsilon$, we can project it to a closed trail in the alignment graph with all vertical edges if the trail is from $G_1$ or horizontal edges if the trail is from $G_2$. The total cost of the projected edges must be greater than or equal to the objective~\eqref{eqn:cctedlemobj}. On the other hand, every collection of trails $C$ that satisfy constraint~\eqref{eqn:cctedlemproj} can be projected to a cover in each of the input graphs, and $\sum_{c\in C} \delta(c) \geq CCTED(G_1, G_2)$. Hence equality holds.\qedhere
\end{proof}

\subsection{The ILP formulation for CCTED}

We show that \gtedilp proposed by~\citet{gted} solves CCTED.
\begin{theorem}
\label{thm:cctedilp}
    Given two input graphs $G_1$ and $G_2$, the optimal objective value of \gtedilp based on $\alngraph(G_1, G_2)$ is equal to $\CCTED(G_1, G_2)$.
\end{theorem}
\begin{proof}
As shown in the proof of Lemma~\ref{lem:ccted}, any pair of edge-disjoint, closed-trail covers in the input graph can be projected to a set of closed trails in $\alngraph(G_1, G_2)$, which satisfied constraints~\eqref{eqn:cycle_cst}-\eqref{eqn:gtedlp_last}. The objective of this feasible solution, which is the total cost of the projected closed trails, equals \CCTED. Therefore, $\CCTED(G_1, G_2)$ is greater than or equal to the objective of \gtedilp.

Conversely, we can transform any feasible solutions of \gtedilp to a pair of covers of $G_1$ and $G_2$. We can do this by transforming one closed trail at a time from the subgraph of the alignment graph, $\alngraph'$ induced by edges with ILP variable $x_{uv} = 1$. Let $c$ be a closed trail in $\alngraph'$. Let $c_1 = \Pi_1(c)$ and $c_2 = \Pi_2(c)$ be two closed trails in $G_1$ and $G_2$ that are projected from $c$. We can construct an alignment between $\str(c_1)$ and $\str(c_2)$ from $c$ by adding match or insertion/deletion columns for each match or insertion/deletion edges in $c$ accordingly. The cost of the alignment is equal to the total cost of edges in $c$ by the construction of the alignment graph. We can then remove edges in $c$ from the alignment graph and edges in $c_1$ and $c_2$ from the input graphs, respectively. The remaining edges in $\alngraph'$ and $G_1$ and $G_2$ still satisfy the constraints~\eqref{eqn:cycle_cst}-\eqref{eqn:gtedlp_last}. Repeat this process and we get a total cost of $\sum_{e\in E} x_{e}\delta(e)$ that aligns pairs of closed trails that form covers of $G_1$ and $G_2$. This total cost is greater than or equal to $\CCTED(G_1, G_2)$. 
\end{proof}

\subsection{CCTED is a lower bound of GTED}

Since the constraints for~\eqref{eqn:lowboundilp} are a subset of~\eqref{eqn:newgtedlp_obj}, a feasible solution to~\eqref{eqn:newgtedlp_obj} is always a feasible solution to~\eqref{eqn:lowboundilp}. Since two ILPs have the same objective function, $\CCTED(G_1, G_2) \leq \GTED(G_1, G_2)$ for any pair of graphs. Moreover, when the solution to~\eqref{eqn:lowboundilp} forms only one connected component, the optimal value of~\eqref{eqn:lowboundilp} is equal to GTED.

\begin{theorem}
\label{thm:cctedgtedequal}
     Let $\alngraph'(G_1, G_2)$ be the subgraph of $\alngraph(G_1, G_2)$ induced by edges $(u,v)\in E$ with $x^{opt}_{uv} = 1$ in the optimal solution to~\eqref{eqn:lowboundilp}. There exists $\alngraph'(G_1, G_2)$ that has exactly one connected component if and only if $c^{opt} = \GTED(G_1, G_2)$. 
\end{theorem}

\begin{proof}
    We first show that if $c^{opt}= \GTED(G_1, G_2)$, then there exists $\alngraph'(G_1, G_2)$ that has one connected component. A feasible solution to~\eqref{eqn:newgtedlp_obj} is always a feasible solution to~\eqref{eqn:lowboundilp}, and since $c^{opt}= \GTED(G_1, G_2)$, an optimal solution to~\eqref{eqn:newgtedlp_obj} is also an optimal solution to~\eqref{eqn:lowboundilp}, which can induce a subgraph in the alignment graph that only contains one connected component.
    
    Conversely, if $x^{opt}$ induces a subgraph in the alignment graph with only one connected component, it satisfies constraints~\eqref{cons:scclinear1}-\eqref{cons:scclinear4} and therefore is feasible to the ILP for GTED~\eqref{eqn:newgtedlp_obj}. Since $c^{opt} \leq \GTED(G_1, G_2)$, this solution must also be optimal for $\GTED(G_1, G_2)$.
\end{proof}
In practice, we may estimate GTED approximately by the solution to~\eqref{eqn:lowboundilp}. As we show in Section~\ref{experiments}, the time needed to solve~\eqref{eqn:lowboundilp} is much less than the time needed to solve GTED. However, in adversarial cases, $c^{opt}$ could be zero but GTED could be arbitrarily large. We can determine if the $c^{opt}$ is a lower bound on GTED or exactly equal to GTED by checking if the subgraph induced by the solution to~\eqref{eqn:lowboundilp} has multiple connected components. 

\subsection{NP-completeness of CCTED}
We prove that the CCTED problem (Problem~\ref{prob:ccted}) is NP-complete by reducing from the Eulerian Trail Equaling Word problem~\citep{eulerw}.
\begin{theorem}
    Computing CCTED is NP-complete.
\end{theorem}

\begin{proof}
   Let Eulerian graph $G=(V, E, \ell, \Sigma)$ and $s$ be an instance of the {\sc Eulerian Tour Equaling Word} problem. Construct two graphs, $G_1$ and $G_2$. If $G$ contains open Eulerian trails, add an edge directing from the sink of the graph to the source of the graph. Let the label of the added edge be $\#$ that does not appear in $\Sigma$. Let the modified graph be $G_1$. If $G$ contains closed Eulerian trails, let $G_1$ be the same as $G$. Let $G_2$ be a graph that contains one cycle with $|E_1|$ edges, where $E_1$ is the edge set of $G_1$. Assign labels to the edges in $G_2$ such that the cycle in $G_2$ spells $s$ if $G$ contains closed Eulerian trails, $s\#$ otherwise.

   If $\CCTED(G_1, G_2) = 0$, $G_2$ must contain at least one closed Eulerian trail that spells some circular permutation of $s\#$. If CCTED is not zero, it means that $s$ must not match Eulerian trails in $G$. 
\end{proof}

\section{Empirical evaluation of the ILP formulations for GTED and its lower bound}
\label{experiments}
\subsection{Implementation of the ILP formulations}
\label{sec:implementation}
We implement the algorithms and ILP formulations for ~\eqref{eqn:newgtedlp_obj}, ~\eqref{eqn:compactgtedilp_obj} and ~\eqref{eqn:lowboundilp}. In practice, the multi-set of edges of each input graph may contain many duplicates of edges that have the same start and end vertices due to repeats in the strings. We reduce the number of variables and constraints in the implemented ILPs by merging the edges that share the same start and end nodes and record the multiplicity of each edge. Each $x$ variable is no longer binary but a non-negative integer that satisfies the modified projection constraints~\eqref{cons:projection}:
\begin{align}
    \sum_{(u,v)\in E} x_{uv}I_{i}((u,v), f) = M_i(f)\quad\text{for all }(u,v)\in E,~f\in G_i, u\neq s, v\neq t,
\end{align}
where $M_i(f)$ is the multiplicity of edge $f$ in $G_i$. Let $C$ be the strongly connected component in the subgraph induced by positive $x_{uv}$, now $\sum_{(u,v)\in E(C)} x_{uv}$ is no longer upper bounded by $|E(C)|$. Therefore, constraints~\eqref{cons:scclinear2} is changed to   
\begin{align}
    &\sum_{(u,v)\in E(C)}x_{uv}-|E(C)|+1-W(C)\beta_{C}\leq 0\quad\text{for all }C\in \mathcal{C}, \\
    &W(C)=\sum_{(u,v)\in E(C)}\max\left(\sum_{f\in G_1}M_1(f)I_{1}((u,v), f),\sum_{f\in G_2}M_2(f)I_{2}((u,v), f)\right),\nonumber
\end{align}
where $W(C)$ is the maximum total multiplicities of edges in the strongly connected subgraph in each input graph that is projected from $C$.

Likewise, constraints~\eqref{cons:orderlinearmid} that set the upper bounds on the ordering variables also need to be modified as the upper bound of the ordering variable $d_{uv}$ for each edge no longer represents the order of one edge but the sum of orders of copies of $(u,v)$ that are selected, which is at most $|E|^2$. Therefore, constraint~\eqref{cons:orderlinearmid} is changed to
\begin{align}
    d_{uv} - |E|^2(1-y_{uv}) \leq 0.
\end{align}
The rest of the constraints remain unchanged.

 We ran all our experiments on a server with 48 cores (96 threads) of Intel(R) Xeon(R) CPU E5-2690 v3 @ 2.60GHz and 378 GB of memory. The system was running Ubuntu 18.04 with Linux kernel 4.15.0. We solve all the ILP formulations and their linear relaxations using the Gurobi solver~\citep{gurobi} using 32 threads.

\subsection{GTED on simulated TCR sequences}
\label{sec:tcr_results}
We construct 20 de Bruijn graphs with $k=4$ using 150-character sequences extracted from the V genes from the IMGT database~\citep{lefranc2011imgt}. We solve the linear relaxation of~\eqref{eqn:compactgtedilp_obj}, \eqref{eqn:newgtedlp_obj} and~\eqref{eqn:lowboundilp} and their linear relaxation on all 190 pairs of graphs. We do not show results for solving~\eqref{eqn:compactgtedilp_obj} for GTED on this set of graphs as the running time exceeds 30 minutes on most pairs of graphs.

To compare the time to solve the ILP formulations when GTED is equal to the optimal objective of~\eqref{eqn:lowboundilp}, we only include 168 out of 190 pairs where GTED is equal to the lower bound (GTED is slightly higher than the lower bound in the remaining 22 pairs). On average, it takes 26 seconds wall-clock time to solve~\eqref{eqn:lowboundilp}, and 71 seconds to solve~\eqref{eqn:newgtedlp_obj} using the iterative algorithm. On average, it takes 9 seconds to solve the LP relaxation of~\eqref{eqn:compactgtedilp_obj} and 1 second to solve the LP relaxation of~\eqref{eqn:lowboundilp}. The time to construct the alignment graph for all pairs is less than 0.2 seconds. The distribution of wall-clock running time is shown in Figure~\ref{fig:gtedilp_runtime}(a). The time to solve~\eqref{eqn:newgtedlp_obj} and~\eqref{eqn:lowboundilp} is generally positively correlated with the GTED values (Figure~\ref{fig:gtedilp_runtime}(b)). On average, it takes 7 iterations for the iterative algorithm to find the optimal solution that induces one strongly connected subgraph (Figure~\ref{fig:gtedilp_runtime}(c)).

In summary, it is fastest to compute the lower bound of GTED. Computing GTED exactly by solving the proposed ILPs on genome graphs of size 150 is already time consuming. When the sizes of the genome graphs are fixed, the time to solve for GTED and its lower bound increases as GTED between the two genome graphs increases. In the case where GTED is equal to its lower bound, the subgraph induced by some optimal solutions of~\eqref{eqn:lowboundilp} contains more than one strongly connected component. Therefore, in order to reconstruct the strings from each input graph that have the smallest edit distance, we generally need to obtain the optimal solution to the ILP for GTED. In all cases, the time to solve the~\eqref{eqn:newgtedlp_obj} is less than the time to solve the~\eqref{eqn:compactgtedilp_obj}.

\begin{figure}
    \centering
    \includegraphics[width=0.95\textwidth]{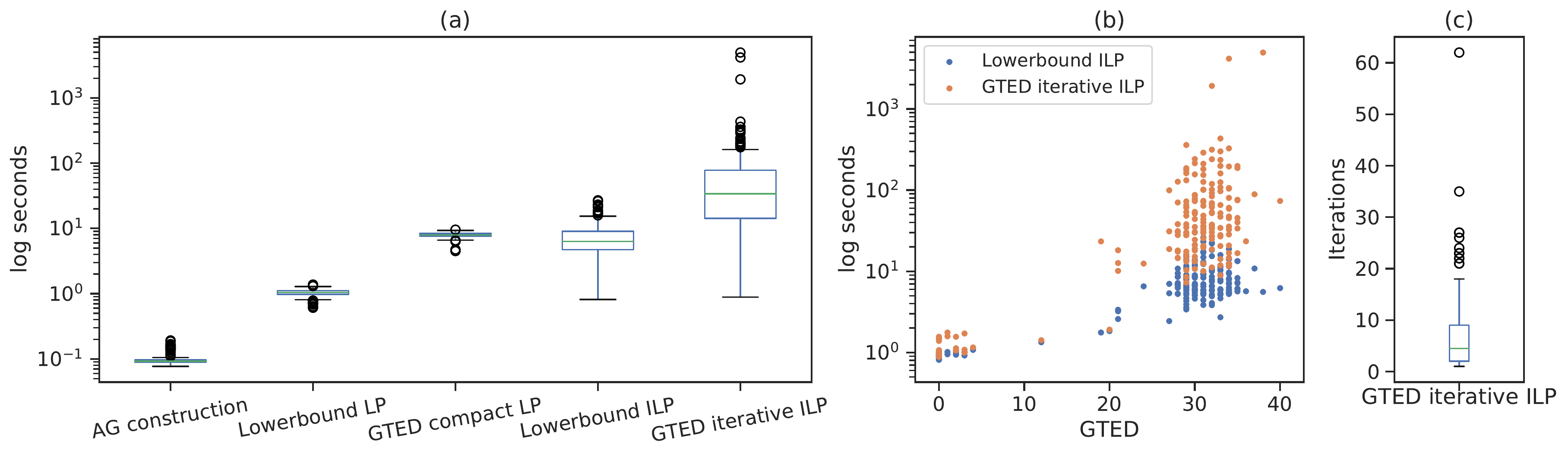}
    \caption{(a) The distribution of wall-clock running time for constructing alignment graphs, solving the ILP formulations for GTED and its lower bound, and their linear relaxations on the log scale. (b) The relationship between the time to solve~\eqref{eqn:lowboundilp},~\eqref{eqn:newgtedlp_obj} iteratively and GTED. (c) The distribution of the number of iterations to solve~\ref{eqn:newgtedlp_obj}. The box plots in each plot show the median (middle line), the first and third quantiles (upper and lower boundaries of the box), the range of data within 1.5 inter-quantile range between Q1 and Q3 (whiskers), and the outlier data points.}
    \label{fig:gtedilp_runtime}
\end{figure}

\subsection{GTED on difficult cases}

Repeats, such as segmental duplications and translocations~\citep{li2020patterns,darai2008segmental} in the genomes increase the complexity of genome comparisons. We simulate such structures with a class of graphs that contain $n$ simple cycles of which $n-1$ peripheral cycles are attached to the $n$-th central cycle at either a node or a set of edges (Figure~\ref{fig:runtime_3cycle}(a)). The input graphs in Figure~\ref{fig:separate_cycle_exp} belong to this class of graphs that contain 2 cycles. This class of graphs simulates the complex structural variants in disease genomes or the differences between genomes of different species.

We generate pairs of 3-cycle graphs with varying sizes and randomly assign letters from \texttt{\{A,T,C,G\}} to edges. We compute the lower bound of GTED and GTED using \eqref{eqn:lowboundilp} and~\eqref{eqn:compactgtedilp_obj}, respectively. We denote the lower bound of GTED computed by solving ~\eqref{eqn:lowboundilp} as $\GTED_l$. We group the generated 3-cycle graph pairs based on the value of $(\GTED- \GTED_l)$ and select 20 pairs of graphs randomly for each $(\GTED- \GTED_l)$ value ranging from 1 to 5. The maximum number of edges in all selected graphs is 32. 

We show the difficulty of computing GTED using the iterative algorithm on the 100 selected pairs of 3-cycle graphs. We terminate the ILP solver after 20 minutes. As shown in Figure~\ref{fig:runtime_3cycle}, as the difference between GTED and $\GTED_l$ increases, the wall-clock time to solve~\eqref{eqn:newgtedlp_obj} for GTED increases faster than the time to solve~\eqref{eqn:compactgtedilp_obj} for GTED. For pairs on graphs with $(\GTED- \GTED_l)= 5$, on average it takes more than 15 minutes to solve~\eqref{eqn:newgtedlp_obj} with more than 500 iterations. On the other hand, it takes an average of 5 seconds to solve~\eqref{eqn:compactgtedilp_obj} for GTED and no more than 1 second to solve for the lower bound. The average time to solve each ILP is shown in Table~\ref{tab:runtime}.

In summary, on the class of 3-cycle graphs introduced above, the difficulty to solve GTED via the iterative algorithm increases rapidly as the gap between GTED and $\GTED_l$ increases. Although~\eqref{eqn:newgtedlp_obj} is solved more quickly than~\eqref{eqn:compactgtedilp_obj} for GTED when the sequences are long and the GTED is equal to $\GTED_l$ (Section~\ref{sec:tcr_results}), \eqref{eqn:compactgtedilp_obj} may be more efficient when the graphs contain overlapping cycles such that the gap between GTED and $\GTED_l$ is larger.

\begin{figure}
    \centering
    \includegraphics[width=0.95\textwidth]{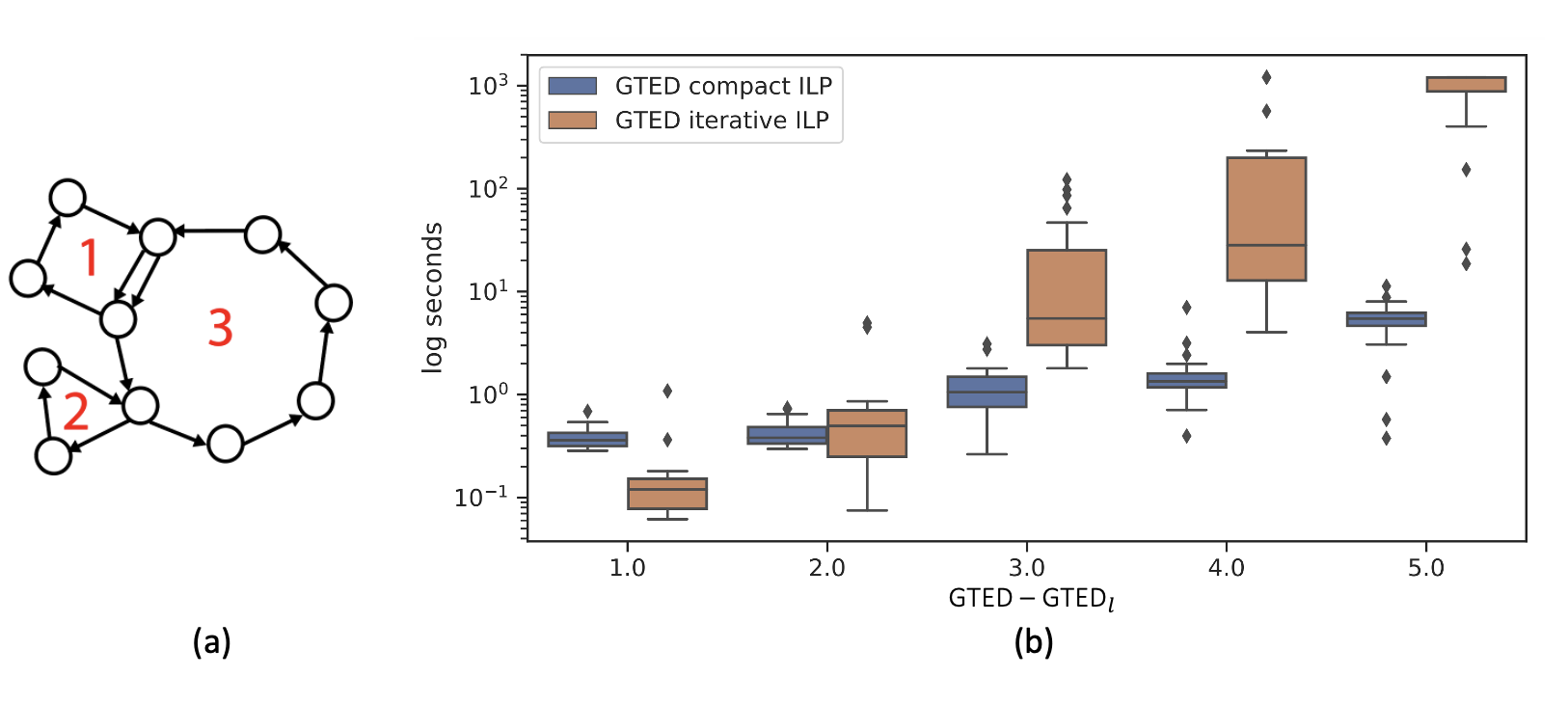}
    \caption{(a) An example of a 3-cycle graph. Cycle 1 and 2 are attached to cycle 3. (b) The distribution of wall-clock time to solve the~\ref{eqn:compactgtedilp_obj} and the iterative~\ref{eqn:newgtedlp_obj} on 100 pairs of 3-cycle graphs.}
    \label{fig:runtime_3cycle}
\end{figure}

\section{Conclusion}

We point out the contradictions in the result on the complexity of labeled graph comparison problems and resolve the contradictions by showing that GTED, as opposed to the results in~\citet{gted}, is NP-complete. On one hand, this makes GTED a less attractive measure for comparing graphs since it is unlikely that there is an efficient algorithm to compute the measure. On the other hand, this result better explains the difficulty of finding a truly efficient algorithm for computing \GTED exactly. In addition, we show that the previously proposed ILP of GTED~\citep{gted} does not solve GTED and give two new ILP formulations of GTED. 

While the previously proposed ILP of GTED does not solve GTED, it solves for a lower bound of GTED, and we show that this lower bound can be interpreted as a more ``local" measure, CCTED, of the distance between labeled graphs. Further, we characterize the LP relaxation of \boundaryilp and show that, contrary to the results in~\citet{gted}, \boundarylp does not always yield optimal integer solutions.

As shown previously~\citep{gted,fgted}, it takes more than 4 hours to solve~\eqref{eqn:lowboundilp} for graphs that represent viral genomes that contain $\approx 3000$ bases with a multi-threaded LP solver. Likewise, we show that computing GTED using either~\eqref{eqn:newgtedlp_obj} or \eqref{eqn:compactgtedilp_obj} is already slow on small genomes, especially on pairs of simulated genomes that are different due to segmental duplications and translations. The empirical results show that it is currently impossible to solve GTED or its lower bound directly using this approach for bacterial- or eukaryotic-sized genomes on modern hardware. The results here should increase the theoretical interest in GTED along the directions of heuristics or approximation algorithms as justified by the NP-hardness of finding GTED.

\section*{Acknowledgements}

The authors would like to thank the members of the Kingsford Group for their helpful comments throughout this project, in particular Guillaume Mar\c{c}ais. 
The authors thank Marina L Knittel, Jacob M Gilbert, and Cenk Sahinalp for their insightful discussion on the NP-completeness of CCTED. This work was supported in part by the US National Science Foundation [DBI-1937540, III-2232121], the US National Institutes of Health [R01HG012470] and by the generosity of Eric and Wendy Schmidt by recommendation of the Schmidt Futures program. \\Conflict of Interest: C.K. is a co-founder of Ocean Genomics, Inc.

\bibliographystyle{unsrtnat}
\bibliography{gtednotp}
\newpage
\begin{appendices}

\section{Proofs for the NP-completeness of GTED}
\subsection{Reduction from ETEW to GTED}
\label{sec:gtedconflict}
We provide below the complete proof for Theorem~\ref{thm:etewp}. 
\begin{customthm}{\ref{thm:etewp}}
    If $\GTED \in \PP$ then $\ETEW \in \PP$.
\end{customthm}
\begin{proof}
    Let $\langle s, G\rangle$ be an instance of \ETEW. Construct a directed, acyclic graph (DAG), $C$, that has only one path. Let the path in $C$ be $P=(e_1,\dots,e_{|s|})$ and the edge label of $e_{i}$ be $s[i]$. Clearly, $C$ is a unidirectional, edge-labeled Eulerian graph, $P$ is the only Eulerian trail in $C$, and $str(P) = s$.
    
    For the graph $G=(V_G,E_G,\ell_G,\Sigma)$ from the \ETEW instance, which may not be unidirectional, create another graph $G'$ that contains all of the nodes and edges in $G$ except the anti-parallel edges. Let $\Sigma_{G'} = \Sigma\cup\{\epsilon\}$, where $\epsilon$ is a character that is not in $\Sigma$. For each pair of anti-parallel edges $(u,v)$ and $(v,u)$ in $G$, add four edges $(u,w_1), (w_1,v), (v,w_2), (w_2,u)$ by introducing new vertices $w_1,w_2$ to $G'$. Let $\ell_{G'}(u,w_1) = \ell_{G}(u, v)$ and $\ell_{G'}(w_2, u) = \ell_G(v, u)$. Let $\ell_{G'}(w_1,v) = \ell_{G'}(v, w_2) = \epsilon$ for every newly introduced vertex. $G'$ has at most twice the number of edges as $G$ and is Eulerian and unidirectional.
    
    Define the cost of changing a character from $a$ to $b$ $\cost(a,b)$ for $a,b\in\Sigma\cup\{-\}$ to be 0 if $a=b$ and 1 otherwise. ``$-$'' is the gap character indicating an insertion or a deletion. Define $\cost(a, \epsilon)$ with $a\in\Sigma$ to be 1. Define $\cost(-, \epsilon)$ to be 0. 
    
    Use the (assumed) polynomial-time algorithm for $\GTED$ to ask whether $\GTED(C, G') \leq 0$ under edit distance $\Sigma$. If yes, then let $(s_1,s_2)$ be the 0-cost alignment of the strings spelled out by the trails in $C$ and $G'$, respectively. The non-gap characters of $s_1$ must spell out $s$ since there is only one Eulerian trail in $C$. Because the alignment cost is $0$, any $-$ (gap) characters in $s_1$ must be aligned with $\epsilon$ characters in $s_2$ and any non-gap characters in $s_1$ must be aligned to the same character in $s_2$. The trail in $G'$ that spells $s_2$ can be transformed to a trail that spells $s_3$ by collapsing the edges with $\epsilon$ character labels, and $s_3 = s_1$.
    
    If $\GTED(C,G')> 0$, $G$ must not contain an Eulerian trail that spells $s$. Otherwise, such a trail could be extended to a trail introducing some $\epsilon$ characters that could be aligned to $s$ with zero cost by aligning gaps with $\epsilon$ characters.
    
    Hence, an (assumed) polynomial-time algorithm for $\GTED$ solves $\ETEW$ in polynomial time.
\end{proof}
\subsection{Reduction from Hamiltonian Path to GTED}
\label{sec:gtednpproof}
We provide below the complete proof for Theorem~\ref{thm:gtednp}. 
\begin{customthm}{\ref{thm:gtednp}}
    GTED is NP-complete.
\end{customthm}
\begin{proof}
    We reduce from the \textsc{Hamiltonian Path} problem, which asks whether a directed, simple graph $G$ contains a path that visits every vertex exactly once. Here simple means no self-loops or parallel edges. Let $\langle G=(V,E)\rangle$ be an instance of \textsc{Hamiltonian Path}, with $n=|V|$ vertices. The reduction is almost identical to that presented in~\citet{eulerw}, and from here until noted later in the proof the argument is identical except for the technicalities introduced to force unidirectionality (and another minor change described later). The first step is to construct the Eulerian closure of $G$, which is defined as $G'=(V',E')$ where
    \begin{equation}
    V' = \{v^{in}, v^{out} : v \in V\} \cup \{w\},
    \end{equation}
    and $E'$ is the union of the following sets of edges and their labels:
    \begin{itemize}
    \item $E_1 = \{(v^{in},  v^{out}) : v \in V\}$,  labeled \texttt{a},
    \item $E_2 = \{(u^{out}, v^{in}) : (u,v) \in E\}$,  labeled \texttt{b},
    \item $E_3 = \{(v^{out}, v^{in}) : v \in V\}$, labeled \texttt{c},
    \item $E_4 = \{(v^{in}, u^{out}) : (u,v) \in E\}$, labeled \texttt{c},
    \item $E_5 = \{(u^{in},w) : u \in V\}$, labeled \texttt{c},
    \item $E_6 = \{(w,u^{in}) : u \in V\}$, labeled \texttt{b}.
    \end{itemize}
    Since $G'$ is connected and every outgoing edge in $G'$ has a corresponding antiparallel incoming edge, $G'$ is Eulerian. It is not unidirectional, so we further create $G''$ from $G'$ by adding dummy nodes to each pair of antiparallel edges and labelling the length-2 paths so created with \texttt{x\#}, where \texttt{x} is the original label of the split edge (\texttt{a}, \texttt{b}, or \texttt{c}) and \texttt{\#} is some new symbol (shared between all the new edges). We call these length-2 paths introduced to achieve unidirectionality ``split edges''.
    
    We now argue that $G$ has a Hamiltonian path iff $G''$ has an Eulerian trail that spells out 
    \begin{equation}
    q = \texttt{a\#}(\texttt{b\#a\#})^{n-1}(\texttt{c\#})^{2n-1}(\texttt{c\#b\#})^{|E| + 1}.
    \end{equation}
    If such an Eulerian trail exists, then the trail starts with spelling the string $\texttt{a\#}(\texttt{b\#a\#})^{n-1}$, which corresponds to a Hamiltonian trail in $G$ since it visits exactly $n$ ``vertex split edges'' (type $E_1$, labeled \texttt{a\#}) and each vertex split edge can be used only once (since it is an Eulerian trail). Further, successively visited vertices must be connected by an edge in $G$ since those are the only \texttt{b\#} split edges in $G''$ (except those leaving $w$, but $w$ must not be involved in spelling out $\texttt{a\#}(\texttt{b\#a\#})^{n-1}$, since entering $w$ requires using a split edge labeled $\texttt{c\#}$).
    
    \newcommand\pred{\text{pred}}
    
    For the other direction,  if a $G$ has a Hamiltonian path $v_1,\dots,v_n$, then walking that sequence of vertices in $G''$ will spell out $\texttt{a\#}(\texttt{b\#a\#})^{n-1}$. This path will cover all $E_1$ edges and the $E_2$ edges that are on the Hamiltonian path. Retracing the path so far in reverse will use $2n-1$ split edges labeled \texttt{c\#}, consuming the $(\texttt{c\#})^{2n-1}$ term in $q$ and covering all nodes' reverse vertex edges $E_3$ (since the path is Hamiltonian). The reverse path also covers the $E_4$ edges corresponding to reverse Hamiltonian path edges. Our Eulerian trail is now ``at'' node $v_1^{in}$.
    
    What remains is to complete the Eulerian walk covering (a)~edges and their antiparallel counterparts corresponding to 
    edges in $G$ that were not used in the Hamiltonian path, and (b)~the edges adjacent to node $w$. To do this, define $
    \pred(v)$ be the vertices $u$ in $G$ for which edge $(u,v)$ exists and $u$ is not the predecessor of $v$ along 
    the Hamiltonian path. For each $u\in\pred(v_1)$, traverse the split 
    edge labeled \texttt{c\#} to $u^{out}$ then traverse the forward split edge labeled $\texttt{b\#}$ back to $v_1^{in}$. This results in a string 
    $(\texttt{c\#b\#})^{|\pred(v_1)|}$. Once the predecessors of $v_1$ are exhausted, traverse the split edge labeled \texttt{c\#} from $v_1^{in}$
    into node $w$ and then traverse the split edge labeled \texttt{b\#} to $v_2^{in}$. This again generates a \texttt{c\#b\#} string. Repeat the process, covering the edges of 
    $v_2$'s predecessors and returning to $w$ to move to the next node along the Hamiltonian path for each node 
    $v_3,\dots,v_n$. After covering the predecessors of $v_n^{in}$, go to $v_1^{in}$ through the remaining edges in $E_5$ and $E_6$, $(v_n^{in},w)$ and $(w, v_1^{in})$, which 
    completes the Eulerian tour. This covers all the edges of $G''$. The word spelled out in this last section of the Eulerian trail is a sequence of repetitions of \texttt{c\#b\#}, with one repetition for 
    each  edge that is not in the Hamiltonian path ($|E|-n+1$) and all of the edges in $E_5$ and $E_6$ for entering and leaving each node ($2n$), with a total of $|E|+1$ repetitions, which is the final $(\texttt{c\#b\#})^{|E| + 1}$ term in $q$.
    
    This ends the slight modification of the proof in~\citet{eulerw}, where the differences are (a)~the introduction of the \texttt{\#} characters and (b)~using the exponent $|E|+1$ of the final part of $q$ instead of $|E|+n+1$ as in~\citet{eulerw} since we create $w$-edges only to $v^{in}$ vertices. (This second change has no material effect on the proof, but reduces the length of the string that must be matched.)
    
    Now, given an instance $\langle G=(V,E)\rangle$ of \textsc{Hamiltonian Path}, with $n=|V|$ vertices, we construct 
    $G''$ as above (obtaining a unidirectional Eulerian graph) and create graph $C$ that only represents string $q$. Note that 
    $|\Sigma|=4$ and $G''$ and $C$ can be constructed in polynomial time. $\GTED(G'',C)= 0$ if and only if an Eulerian path in $G''$ spells out $q$, since there can be no indels or mismatches. By the above argument, an 
    An eulerian tour that spells out $q$ exists if and only if $G$ has a Hamiltonian path.
\end{proof}

\subsection{FGTED is NP-complete}
\label{sec:fgtednp}
\begin{problem}[Flow Graph Traversal Edit Distance (\FGTED)~\citep{fgted}] 
Given unidirectional, edge-labeled Eulerian graphs $G_1$ and $G_2$, each of which has distinguished $s_1,s_2$ source and $t_1,t_2$ sink vertices, compute
\begin{equation}
\FGTED(G_1,G_2) \triangleq \min_{\substack{D_1\in \flow(G_1, s_1, t_1)\\D_1\in \flow(G_2, s_2, t_2)}} \emedit(\strset(D_1), \strset(D_2)),
\end{equation}
where $\flow(G_i, s_i, t_i)$ is the collection of all possible sets of $s_1$-$t_1$ trail decomposition of saturating flow from $s_i$ to $t_i$, $\strset(D)$ is the multi-set of strings constructed from trails in $D$.
\end{problem}

\begin{customthm}{\ref{thm:fgtednp}}
    FGTED is NP-complete.
\end{customthm}
\begin{proof}
        Let $G = (v, E)$ be an instance of the {\sc Hamiltonian Cycle} problem. Let $n = |V|$ be the number of vertices in $G$. Construct the Eulerian closure of $G$ and split the anti-parallel edges. Let the new graph be $G' = (V', E')$. Attach a source $s$ and a sink node $t$ to an arbitrary node $v_1^{in}$ by adding edge $(s, v_1^{in})$ and $(v_1^{in}, t)$ with labels \texttt{s} and \texttt{t}, respectively. 
        
        Construct a string $q$, such that 
        \begin{align}
            q=\texttt{sa\#}(\texttt{b\#a\#})^{n-1}(\texttt{c\#})^{2n-1}(\texttt{c\#b\#})^{|E| + 1}\texttt{t}.
        \end{align}
        Create a graph $Q$ that only contains one path with labels on the edges of the path that spell the string $q$. 
       The union of the set of trails in any flow decomposition of $G'$ is equal to a set of Eulerian trails, $\mathcal{E}$, that starts at $s$ and ends at $t$. All Eulerian trails in $\mathcal{E}$ are also closed Eulerian trails of $G'\setminus \{s,t\}$ that starts and ends at $v_1^{in}$.
    
        Using the same line of argument in the proof of Theorem~\ref{thm:gtednp}, an Eulerian trail in $G'$ that spells $q$ is equivalent to a Hamilton Cycle in $G$. In addition, $\FGTED(Q, G') = 0$ if and only if all Eulerian trails in $\mathcal{E}$ spell out $q$. Therefore, if $\FGTED(Q, G') = 0$, then there is a Hamiltonian Cycle in $G$. Otherwise, then there must not exist a Hamiltonian Cycle in $G$.
    \end{proof}

\section{Equivalence between two ILPs proposed by~\citeauthor{gted}}
\label{sec:ilpequivalence}

The analysis provided by~\citet{gted} states that the LP relaxation of \gtedilp does not always yield integer solutions, but the LP relaxation of \boundaryilp always yields integer solutions. This suggests that the two LP relaxations have difference feasibility regions for $x$. We show that these two LP relaxations are actually equivalent in Theorem~\ref{thm:lpequal}. Further, we show that \gtedilp and \boundaryilp are also equivalent. Since \gtedilp does not solve for $\GTED(G_1, G_2)$ as shown in~\ref{sec:gtedilpcounterexp}, we conclude that \boundaryilp also does not solve $\GTED(G_1, G_2)$.

\begin{theorem}\label{thm:lpequal}
    Given two unidirectional, edge-labeled Eulerian graphs $G_1$, $G_2$, the feasibility region of $x$ in the LP relaxation of \boundaryilp is the same as the feasibility region of $x$ in the LP relaxation of \gtedilp. 
\end{theorem}

Let $\alngraph (G_1,G_2)=(V,E,\delta)$ be the alignment graph of $G_1=(V_1, E_1, \ell_1, \Sigma_1)$ and $G_2=(V_2, E_2, \ell_2, \Sigma_2)$, and let $T(G_1,G_2)$ be its two-simplex set. First, we have the following result:
\begin{lemma}
    Let $[y_{i}]\in\mathbb{R}^{|T(G_1,G_2)|}$ be a vector such that the $j$-th entry of $[y_{i}]$, $[y_{i}]_{j}$ is equal to $0$ for all $j\neq i$. The vector $x'=x+[\partial][y_{i}]$ satisfies the constraints~\eqref{eqn:cycle_cst}-\eqref{eqn:proj_cst} if the vector $x$ satisfies the constraints~\eqref{eqn:cycle_cst}-\eqref{eqn:proj_cst}. 
    \label{lem:eqonedir}
\end{lemma}

\begin{proof}
    Let $\sigma_{i}\in T(G_1,G_2)$ be the 2-simplex corresponding to the entry $i$ of $[y_{i}]$. Based on the construction of $T(G_1,G_2)$, $\sigma_i$ has two forms: $[(u_1,u_2),(v_1,u_2), (v_1,v_2)]$ or $[(u_1,u_2),(u_1,v_2), (v_1,v_2)]$. Without loss of generality, we assume $\sigma_i=[(u_1,u_2),(v_1,u_2), (v_1,v_2)]$. We can prove this lemma by using the same way when $\sigma_i=[(u_1,u_2),(u_1,v_2), (v_1,v_2)]$. 
    Since 
   \begin{align*}
       \partial \sigma_i=[(u_1,u_2),(v_1,u_2)]+[(v_1,u_2),(v_1,v_2)]-[(u_1,u_2),(v_1,v_2)],
   \end{align*} 
   We have
    \begin{align*}
        [\partial][y_i]=[y_i]_i[x_{e_1}]+[y_i]_i[x_{e_2}]-[y_i]_i[x_{e_3}],
    \end{align*}
where $e_1=[(u_1,u_2),(v_1,u_2)]$, $e_2=[(v_1,u_2),(v_1,v_2)]$, $e_3=[(u_1,u_2),(v_1,v_2)]$, and $[x_e]\in\mathbb{R}^{|E|}$ is a vector such that all the entries are $0$ except that the one corresponding to edge $e$ is $1$. we also let $[x_v]\in\mathbb{R}^{|V|}$ be a vector such that all the entries are $0$ except that the one corresponding to vertex $v$ is $1$. Therefore, we have  
\begin{align*}
    Ax' = Ax+[y_i]_i[x_{v_2}]-[y_i]_i[x_{v_1}]+[y_i]_i[x_{v_3}]-[y_i]_i[x_{v_2}]-[y_i]_i[x_{v_3}]+[y_i]_i[x_{v_1}]=Ax,
\end{align*}
where $v_1=(u_1,u_2)$, $v_2=(v_1,u_2)$, and $v_3=(v_1,v_2)$.
Hence, $x'$ satisfies the constraint~\eqref{eqn:cycle_cst} if $x$ satisfies the constraint~\eqref{eqn:cycle_cst}. 

In addition, since $\sum_{e\in E} x_e'I_i(e, f)=\sum_{e\in E} x_eI_i(e, f)+[y_i]_i I_i(e_1, f)+[y_i]_i I_i(e_2, f)-[y_i]_i I_i(e_3, f)$, and:
   \begin{itemize}
       \item $I_1(e_1,(u_1,v_1))=1$ and $I_i(e_1,f)=0$ for other $f\in G_i$, 
       \item $I_2(e_2,(u_2,v_2))=1$ and $I_i(e_2,f)=0$ for other $f\in G_i$, 
       \item $I_1(e_3,(u_1,v_1))=1$, $I_2(e_3,(u_2,v_2))=1$, and  $I_i(e_3,f)=0$ for other $f\in G_i$, 
   \end{itemize}
   we have:
   \begin{itemize}
       \item $[y_i]_i I_1(e_1,(u_1,v_1))+[y_i]_i I_1(e_2,(u_1,v_1))-[y_i]_i I_1(e_3,(u_1,v_1))=[y_i]_i+0-[y_i]_i=0$,
       \item $[y_i]_i I_2(e_1,(u_2,v_2))+[y_i]_i I_2(e_2,(u_2,v_2))-[y_i]_i I_2(e_3,(u_2,v_2))=0+[y_i]_i-[y_i]_i=0$,
       \item $[y_i]_i I_i(e_1,f)+[y_i]_i I_i(e_2,f)-[y_i]_i I_i(e_3,f)=0+0-0=0$ for any other $i=1,2$ and $f\in E_i$. 
   \end{itemize}
Therefore, $\sum_{e\in E} x_e'I_i(e, f)=\sum_{e\in E} x_eI_i(e, f)$, meaning that $x'$ satisfies the constraint~\eqref{eqn:proj_cst} if $x$ satisfies the constraint~\eqref{eqn:proj_cst}.
\end{proof}
With Lemma~\ref{lem:eqonedir}, we prove that any feasible solution of $x$ in~\eqref{eqn:boundarylp} is a feasible solution of~\eqref{eqn:gtedlp_obj}-\eqref{eqn:gtedlp_last}. First, it is easy to check that $x^{init}$ satisfies the constraints~\eqref{eqn:cycle_cst}-\eqref{eqn:proj_cst}. For each feasible solution of $x$ in~\eqref{eqn:boundarylp}, since $x=x^{init}+[\partial]y=x^{init}+\sum_{i}[\partial][y_{i}]$, by iteratively using Lemma~\ref{lem:eqonedir}, we get that $x$ satisfies the constraints~\eqref{eqn:cycle_cst}-\eqref{eqn:proj_cst}. Since $x_e\geq 0$ for all $e\in E$ is a constraint existing in both linear relaxations, $x$ is a feasible solution of~\eqref{eqn:gtedlp_obj}-\eqref{eqn:gtedlp_last}. 


We now show that any feasible solution of~\eqref{eqn:gtedlp_obj}-\eqref{eqn:gtedlp_last} is a feasible solution of~\eqref{eqn:boundarylp}. Let $x$ be a feasible solution of~\eqref{eqn:gtedlp_obj}-\eqref{eqn:gtedlp_last}. We show that $x$ is also a feasible solution of~\eqref{eqn:boundarylp} by proving that $x$ can be converted to $x^{init}$ in~\eqref{eqn:boundarylp} via the boundary operator $\partial$. First, if there is a diagonal edge $e=[(u_1,u_2),(v_1,v_2)]$ in $E$ such that $x_e>0$, then it can be replaced by the horizontal edge $e_h=[(u_1,u_2),(u_1,v_2)]$ followed by the vertical edge $e_v=[(u_1,v_2),(v_1,v_2)]$ by using one boundary operation on the 2-simplex $[(u_1,u_2),(u_1,v_2),(v_1,v_2)]$. Hence, $x$ can be converted to a new vector $x'$, such that $x'_{e}=0$, $x'_{e_{h}}=x_{e_{h}}+x_{e}$, $x'_{e_{v}}=x_{e_{v}}+x_{e}$, and all the other entries in $x'$ are the same as those in $x$. It is easy to check that $x'$ is also a feasible solution of~\eqref{eqn:gtedlp_obj}-\eqref{eqn:gtedlp_last}. Therefore, without loss of generality, we assume $x$ to be a vector such that all the entries corresponding to diagonal edges in $\alngraph(G_1,G_2)$ are zero. 

We then prove that any $x$ can be converted to $x^{init}$ in~\eqref{eqn:boundarylp} via the boundary operator. Let the source and the sink node of $x$ in $\alngraph (G_1,G_2)$ be $(s_1^{1},s_1^{2})$ and $(s_2^{1},s_2^{2})$, where $s_1^{i}$ is the source node of $G_i$ and $s_2^{i}$ is the sink node of $G_i$. When the Eulerian trail is closed (meaning that it is an Eulerian tour) in $G_i$, we let $s_1^{i}=s_2^{i}$ be an arbitrary vertex in $V_i$. $x^{init}$ can be seen as a trail (tour) in $\alngraph(G_1,G_2)$ that starts from $(s_1^{1},s_1^{2})$, walks along an Eulerian trail of $G_2$ via all the horizontal edges $P_h$, 
\begin{equation*}
    P_h=\{[(s_1^{1},s_1^{2}),(s_1^{1},v_1^{2})],[(s_1^{1},v_1^{2}),(s_1^{1},v_2^{2})],\dots,[(s_1^{1},v_{i-1}^{2}),(s_1^{1},v_{i}^{2})],
[(s_1^{1},v_{i}^{2}),(s_1^{1},s_{2}^{2})]\},
\end{equation*}
and then walks along an Eulerian trail of $G_1$ via all the vertical edges $P_v$, 
\begin{equation*}
P_v=\{[(s_1^{1},s_2^{2}),(v_1^{1},s_2^{2})], [(v_1^{1},s_2^{2}),(v_2^{1},s_2^{2})],\dots,[(v_{j-1}^{1},s_2^{2}),(v_j^{1},s_2^{2})], [(v_j^{1},s_2^{2}),(s_2^{1},s_2^{2})]\},
\end{equation*}
until the sink node $(s_2^{1},s_2^{2})$. Here $\{s_1^{2},v_1^{2},v_2^{2},\dots, v_{i-1}^{2}, v_{i}^{2}, s_{2}^{2}\}$ is an Eulerian trail of $G_2$ and $\{s_1^{1},v_1^{1},v_2^{1},\dots, v_{i-1}^{1}, v_{i}^{1}, s_{2}^{1}\}$ is an Eulerian trail of $G_1$. We use $P_0=\{P_h,P_v\}$ to denote the trail from $(s_1^{1},s_1^{2})$ to $(s_2^{1},s_2^{2})$ that is the concatenation of $P_h$ and $P_v$. It is easy to see that each edge in $P_0$ is unique. 

As shown in~\citet{fgted}, $x$ is a flow of $\alngraph (G_1,G_2)$ with the additional constraint~\eqref{eqn:proj_cst}. Therefore, according to the flow decomposition theorem~\citep[p.~80]{ahujia1993network}, $x$ can be decomposed into a finite set of weighted paths in $\alngraph(G_1,G_2)$ from $(s_1^{1},s_1^{2})$ to $(s_2^{1},s_2^{2})$, which is denoted as $\{(p_1,w_1^{p}),\dots, (p_n,w_n^{p})\}$, and a finite set of weight cycles in $\alngraph(G_1,G_2)$, which is denoted as $\{(c_1,w_1^{c}),\dots, (c_m,w_m^{c})\}$. Each path or cycle only contains horizontal and vertical edges.

 For path $i$, we use a vector $x^{p,i}$ to represent $(p_i,w_i^{p})$, 
    \begin{align}
    x^{p,i}_e = \begin{cases}
    w_i^{p}\quad &\text{if}~e\in p_i\\
    0\quad &\text{otherwise},
    \end{cases}
\end{align}

\begin{figure}
    \centering
    \includegraphics[width=0.7\columnwidth]{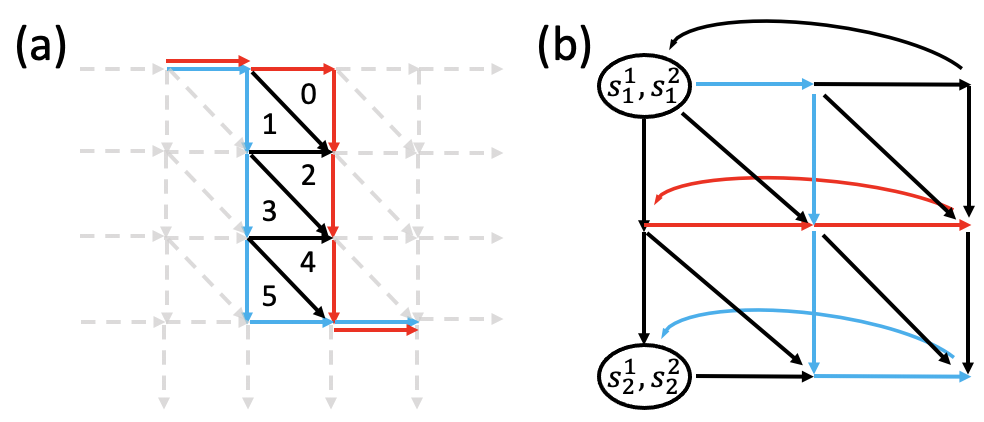}
    \caption{(a) An example of converting three vertical edges followed by one horizontal edge (blue line) to one horizontal edge followed by three vertical edges (red line). It can be done by doing boundary operations on 2-simplices labeled from $0$ to $5$. (b) An example of a cycle path (red line) and its auxiliary trail (blue line). }
    \label{fig:lpequal}
\end{figure}

By using the boundary operator, each path $p_i$ can actually be converted to a new trail $p_i'$ such that each edge in $p_i'$ is also an edge in $P_0$. To prove this, we consider the following two cases:
\begin{itemize}
    \item If $p_i$ walks along all the horizontal edges followed by all the vertical edges, then every edge in $p_i$ is an edge in $P_0$. To see that, let $e$ be an horizontal edge in $p_i$, since $p_i$ starts from $(s_1^{1},s_1^{2})$, $e$ has the form $[(s_1^{1},v),(s_1^{1},v')]$ where $[v,v']\in E_2$. Since $P_h$ corresponds to the Eulerian trail of $G_2$, for each $[v,v']\in E_2$, we have $[(s_1^{1},v),(s_1^{1},v')]\in P_h$. Therefore $e\in P_0$. We can use the same way to prove $e\in P_0$ when $e$ is a vertical edge. Note that in this case, the number of horizontal edges or vertical edges can be zero. 
    \item If not, then we let $p_i=\{e^i_1,e^i_2,\dots , e^i_{m}\}$, and let $e^i_{t}$ be the vertical edge with the smallest index $t$. There exists an integer $k~(k\geq 1)$ such that $\{e^{i}_{t},e^{i}_{t+1},\dots , e^{i}_{t+k-1}\}$ are all vertical edges and $e^{i}_{t+k}$ is an horizontal edge. We denote each vertical edge $e_{t+w}^{i}\in \{e^{i}_{t},e^{i}_{t+1},\dots , e^{i}_{t+k-1}\}$ as $[(v_w,v_t),(v_{w+1},v_t)]$ and denote $e^{i}_{t+k}$ as $[(v_k,v_t),(v_k,v_{t+1})]$. It is easy to see that when $w=0$, $v_w=s_1^{1}$. By using the boundary operator, this subpath $\{e^{i}_{t},e^{i}_{t+1},\dots , e^{i}_{t+k-1},e^{i}_{t+k}\}$ can be replaced by another subpath with one horizontal edge $[(s_1^{1},v_t),(s_1^{1},v_{t+1})]$ followed by $k$ vertical edges: $$\{[(s_1^{1},v_{t+1}),(v_1,v_{t+1})],[(v_1,v_{t+1}),(v_2,v_{t+1})],\dots , [(v_{k-1},v_{t+1}),(v_k,v_{t+1})].$$ Now we have a new path, denoted as $p_i^{1}$, in which the smallest index of the vertical edges becomes $t+1$. Figure~\ref{fig:lpequal}(a) shows an example, in which the blue line represents the subpath of $p_i$ and the red line represents the new subpath in $p_i^{1}$. 
    
    To create a new vector that represents $p_i^{1}$, we first create a zero vector $y^{p,i,1}\in\mathbb{R}^{|T(G_1,G_2)|}$, and from $w=0$ to $w=k-1$, we iteratively update $y^{p,i,1}$ via the following equations:
    \begin{align}
    y^{p,i,1}_\sigma = \begin{cases}
    y^{p,i,1}_\sigma-w_i^{p}\quad &\text{if}~\sigma = [(v_w,v_t),(v_{w+1},v_t),(v_{w+1},v_{t+1})] \\
     y^{p,i,1}_\sigma+w_i^{p}\quad &\text{if}~\sigma = [(v_w,v_t),(v_{w},v_{t+1}),(v_{w+1},v_{t+1})] \\
    0\quad &\text{otherwise}.
    \end{cases}
\end{align}
The vector $x^{p,i,1}=x^{p,i}+[\partial]y^{p,i,1}$ is the one that represents $p_{i}^{1}$.

Since the length of $p_i$ is finite, by doing such a transformation a finite number of times, we can convert $p_i$ to a new path $p_i'$ such that $p_i'$ walks along all the horizontal edges first followed by all the vertical edges, therefore each edge in $p_i'$ is also an edge in $P_0$. We use the vector $\hat{x}^{p,i}$ to represent $p_i'$, $\hat{x}^{p,i}=x^{p,i}+[\partial]\sum_{j=1}^{q}y^{p,i,j}$ where $q$ is the number of transformations. Apperantly, $\hat{x}^{p,i}_e=0$ when $e\notin P_0$. Let $y^{p,i}=\sum_{j=1}^{q}y^{p,i,j}$, we have $\hat{x}^{p,i}=x^{p,i}+[\partial]y^{p,i}$. 
\end{itemize}

For cycle $i$, we also use a vector $x^{c,i}$ to represent $(c_i,w_i^{c})$, 
    \begin{align}
    x^{c,i}_e = \begin{cases}
    w_i^{c}\quad &\text{if}~e\in c_i\\
    0\quad &\text{otherwise},
    \end{cases}
\end{align}
Let $(v,v')$ be an arbitrary chosen node in $c_i$, 
we construct a trail $p_{aux}^{i}$ that passes $(v,v')$ as follows: 
\begin{itemize}
    \item From $(s_1^{1},s_1^{2})$, walk along $P_{h}$ until the node $(s_1^{1},v')$. It corresponds to a part of an Eulerian trail of $G_2$. 
    \item From $(s_1^{1},v')$, walk along an Eulerian trail of $G_1$ to $(s_2^{1},v')$. It must passes the node $(v,v')$. 
    \item From $(s_2^{1},v')$, walk along the remaining part of the Eulerian trail of $G_2$ to the node $(s_2^{1},s_2^{2})$. 
\end{itemize}
Figure~\ref{fig:lpequal}(b) shows an example, in which the blue line represents $p_{aux}^{i}$ and the red line represents $c_i$. 

We use $x^{aux,i}$ to denote the vector representing $p_{aux}^{i}$. The combination of $c_i$ and $p_{aux}^{i}$, represented by the vector $x^{c,i}+x^{aux,i}$ creates a new trail (may have repeated edges) from $(s_1^{1},s_1^{2})$ to $(s_2^{1},s_2^{2})$: (1) walk along $p_{aux}^{i}$ from $(s_1^{1},s_1^{2})$ to $(v,v')$, (2) walk along $c_i$ from $(v,v')$ to itself, and (3) walk along the remaining part of $p_{aux}^{i}$ from $(v,v')$ to $(s_2^{1},s_2^{2})$. By using the same way as we described above, each $c_i + p_{aux}^{i}$ or $p_{aux}^{i}$ can be converted to a new trail in which each edge is also an edge in $P_0$. We use $\hat{x}^{c,i}$ or $\hat{x}^{aux,i}$ to represent the new trail accordingly, therefore, we have  $\hat{x}^{c,i}=x^{c,i}+x^{aux,i}+[\partial]y^{c,i}$ and $\hat{x}^{aux,i}=x^{aux,i}+[\partial]y^{aux,i}$. Likewise, $\hat{x}^{c,i}_{e}=\hat{x}^{aux,i}_{e}=0$ when $e\notin P_0$.

We define a new vector $\hat{x}$ such that:
\begin{align*}
    \hat{x}&=\sum_{i=1}^{n}\hat{x}^{p,i}+\sum_{j=1}^{m}\hat{x}^{c,j}-\hat{x}^{aux,j} \\
    &=\sum_{i=1}^{n}x^{p,i}+[\partial]y^{p,i}+\sum_{j=1}^{m}x^{c,j}+x^{aux,j}+[\partial]y^{c,j}-\left(\sum_{j=1}^{m}x^{aux,j}+[\partial]y^{aux,j}\right) \\
    &=\sum_{i=1}^{n}x^{p,i} + \sum_{j=1}^{m}x^{c,j} + [\partial]\left(\sum_{i=1}^{n}y^{p,i}+\sum_{j=1}^{m}y^{c,j}-\sum_{j=1}^{m}y^{aux,j}\right) \\
    &=x+[\partial]\left(\sum_{i=1}^{n}y^{p,i}+\sum_{j=1}^{m}y^{c,j}-\sum_{j=1}^{m}y^{aux,j}\right).
\end{align*}
Therefore, $\hat{x}$ is a vector converted from $x$ via boundary operations. $\hat{x}$ is equal to $x^{init}$ because:
\begin{enumerate}
    \item $\hat{x}_{e}=0$ when $e\notin P_0$ since $\hat{x}^{p,i}_{e}=\hat{x}^{c,i}_{e}=\hat{x}^{aux,i}_{e}=0$ when $e\notin P_0$ for each $i$. 
    \item As we have proved above, the boundary operator preserves the constraints~\eqref{eqn:cycle_cst}-\eqref{eqn:proj_cst}. Therefore, $\hat{x}$ satisfies the constraints~\eqref{eqn:cycle_cst}-\eqref{eqn:proj_cst} since $x$ is a feasible solution of~\eqref{eqn:gtedlp_obj}-\eqref{eqn:gtedlp_last}. Combined with the first point, we have that  $\hat{x}_{e}=1$ if $e\in P_0$ and  $\hat{x}_{e}=0$ otherwise, meaning that $\hat{x}=x^{init}$.
\end{enumerate}

Hence, for each feasible solution $x$ of~\eqref{eqn:gtedlp_obj}-\eqref{eqn:gtedlp_last}, we have:
\begin{align*}
    x &=x^{init}-[\partial]\left(\sum_{i=1}^{n}y^{p,i}+\sum_{j=1}^{m}y^{c,j}-\sum_{j=1}^{m}y^{aux,j}\right) \\
    &=x^{init}+[\partial]\left(-\sum_{i=1}^{n}y^{p,i}-\sum_{j=1}^{m}y^{c,j}+\sum_{j=1}^{m}y^{aux,j}\right),
\end{align*}
meaning that $x$ is also a feasible solution of~\eqref{eqn:boundarylp}. 

We proved that the feasibility region of $x$ in~\eqref{eqn:boundarylp} is the same as the feasibility region of $x$ in~\eqref{eqn:gtedlp_obj}-\eqref{eqn:gtedlp_last}, and since the objective functions of these two linear relaxations are the same, the optimal solutions of them are equal. 

By employing the same approach and taking into account that if all edge weights in a flow network are non-negative integers, the flow decomposition theorem guarantees that the network can be decomposed into a finite set of weighted paths and cycles, each with positive integer weight, we can prove that \gtedilp and \boundaryilp are also equivalent.

Based on the proof, we can conclude that the way to index the vertices or edges in the alignment graph, or the 2-simplices in $T(G_1,G_2)$, will not affect the equivalence result. Additionally, different choices of orientations for the 2-simplices in $T(G_1,G_2)$ will also not impact the equivalence result. This is because for any two sets $T(G_1,G_2)$ and $T'(G_1,G_2)$ containing the same 2-simplices with the same indices but different orientations, if $(x,y)$ is a feasible solution of~\boundaryilp (or its relaxation) that corresponds to $T(G_1,G_2)$, then $(x,y')$ is a feasible solution of~\boundaryilp (or its relaxation) that corresponds to $T'(G_1,G_2)$, where $y_i=y'_i$ when $\sigma_i\in T(G_1,G_2)$ has the same orientation as $\sigma'_i\in T'(G_1,G_2)$, and $y_i=-y'_i$ when $\sigma_i\in T(G_1,G_2)$ has the opposite orientation to $\sigma'_i\in T'(G_1,G_2)$. Therefore, it is acceptable to specify a particular orientation for each 2-simplex when defining $T(G_1,G_2)$.

\section{The linear relaxation of \boundaryilp does not always yield integer solutions}
\label{tu}
\subsection{$\boundarym$ is not necessarily totally unimodular}
\label{sec:nottu}

\begin{figure*}
    \centering
    \includegraphics[width=\textwidth]{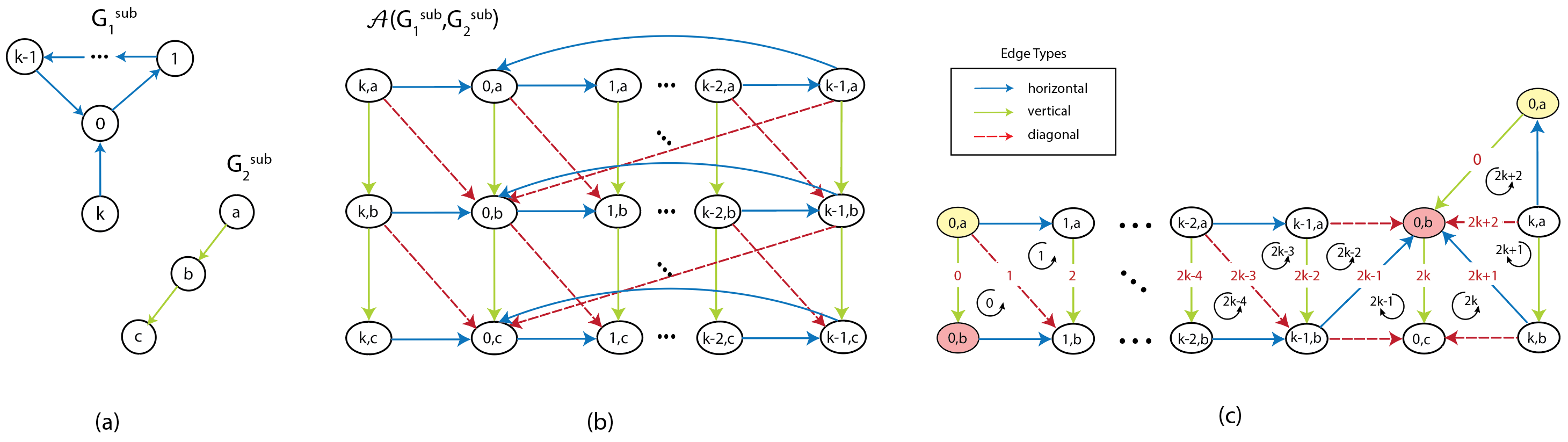}
    \caption{(a) Subgraphs $G_1^{sub}$ and $G_2^{sub}$ of input graphs $G_1$ and $G_2$. Dots represent a path from node 1 to $k-1$ with middle nodes omitted. (b) The alignment graph $\alngraph(G_1^{sub}, G_2^{sub})$ with different edges labeled with colors. (c) A subgraph of the alignment graph in (b) with edges and triangles numbered. Dots represent horizontal and diagonal edges omitted. The same vertices that are repeated in (c) are marked with yellow and red filling colors.}
    \label{fig:alngraph}
\end{figure*}
A linear programming formulation always yields integer solutions if its constraint matrix is totally unimodular, which means that all of its square submatrices have determinants of 0, -1 or 1~\citep{dey2010optimal}. To show that the constraint matrix of the LP relaxation of \boundaryilp is not totally unimodular, we first write the LP in standard form.

In a standard form of a LP, all variables are greater than, or equal to 0. Since $y$ vectors in the LP relaxation of \boundaryilp can contain negative entries, we decompose it into $y^+ - y^-$.  Given alignment graph $\alngraph(G_1, G_2) = (V, E, \delta)$ and $T(G_1, G_2)$, 
we can now write the standard form of \boundarylp as
\begin{equation}
\begin{aligned}
    \underset{x\in \mathbb{R}^{|E|},y^+,y^-\in \mathbb{R}^{|T(G_1,G_2)|}}{\textrm{minimize}}\quad & \sum_{e \in E} x_e\delta(e)\\
    \textrm{subject to}&\quad \left[I,-[\partial],[\partial]\right][x,y^{+},y^{-}]^{\top} = x^{init}\\
    &\quad x,y^{+},y^{-} \geq 0.
\end{aligned}    
\label{eqn:standaradboundarylp}
\end{equation}
Hence the constraint matrix of the LP relaxation is $A = \left[I, -[\partial], [\partial]\right]$. According to the characteristics of a totally unimodular matrix~\citep[p.~280]{schrijver1998theory} $A$ is not totally unimodular if $[\partial]$ is not totally unimodular. We show that $[\partial]$ is not TU when the input graphs satisfy the constraints given in the following theorem.




\begin{theorem} Given two unidirectional, edge-labeled Eulerian graphs $G_1$ and $G_2$ where $|E_1| \geq 2$ and $|E_2| \geq 2$, the boundary matrix $\boundarym$ constructed from $\alngraph (G_1,G_2)=(V,E,\delta)$ and $T(G_1,G_2)$ is not totally unimodular if there is a vertex $v\in V_1 \text{~or~} V_2$ such that there are at least 3 unique edges in $E_1$ or $E_2$ that are incident to $v$. Here, unique edges are edges that connect to $v$ at one end but have different endpoints at the other end.
\label{thm:nontu}
\end{theorem}

\begin{proof}

To prove that the boundary matrix is not TU, we only need to show that it is not TU under one specific chosen orientation for 1- and 2-simplices, as well as one specific chosen set of indices for 1- and 2-simplices. This is because changing the orientations or indices of 1-simplices in $E$ or 2-simplices in $T(G_1,G_2)$ corresponds to permuting rows and columns of $[\partial]$ or multiplying rows and columns of $[\partial]$ by $-1$, which preserves the total unimodularity~\citep[p. 280]{schrijver1998theory}.
    

    Without loss of generality, let $v_0\in V_1$ be a node that is incident to at least 3 unique edges. Since $G_1$ is an Eulerian graph, $v$ must be part of a cycle $C$ in $G_1$. Also, there must exist another node $v_k$ and an edge between $v_0$ and $v_k$ in either direction, such that the edge between $v_0$ and $v_k$ is not contained in cycle $C$ (Figure~\ref{fig:alngraph}(a)). Suppose the number of nodes in the cycle is $k$ ($k\geq 3$ due to the unidirectionality constraint), and let the cycle $C= v_0, v_1,\dots,v_{k-1}$. Since a specific choice of 1-simplex orientations does not affect the total unimodularity of the boundary matrix, we assume the edge between $v_0$ and $v_k$ is $[v_k, v_0]$ without loss of generality. We use $G_1^{sub}=(V_1^{sub}, E_1^{sub})$ to denote the subgraph with $V_1^{sub}=\{v_0,\dots,v_{k-1},v_k\}$ and  $E_1^{sub}=\left\{[v_i,v_{i+1}]: i\in\{0,1,\dots k-2\}\right\}\cup\{[v_k,v_0]\}$. 
    Since $|E_2|\geq 2$ and $G_2$ is a connected graph, there exist two consecutive, directed edges in $G_2$. We use $G_2^{sub}=(V_2^{sub}, E_2^{sub})$ to denote the subgraph of $G_2$ with $V_2^{sub}=\{v_a,v_b,v_c\}$ and $E_2^{sub}=\{[v_a,v_b],[v_b,v_c]\}$. The alignment graph $\alngraph(G_1^{sub},G_2^{sub})$ is formed with $G_1^{sub}$ and $G_2^{sub}$ and is a subgraph of $\alngraph(G_1, G_2)$, therefore, each subgraph of $\alngraph(G_1^{sub},G_2^{sub})$ is also a subgraph of $\alngraph(G_1, G_2)$. Similarly, the 2-simplex set $T(G_1^{sub},G_2^{sub})$ is a subset of $T(G_1,G_2)$. 

    We extract a sequence of 2-simplices (Figure~\ref{fig:alngraph}(c)), $T_c$, from $\Tr(G_1^{sub}, G_2^{sub})$ via following steps:
    \begin{enumerate}
        \item Extract all oriented 2-simplices $[(v_i,v_a),(v_i,v_b),(v_{i+1},v_b)]$ and \\$[(v_i,v_a),(v_{i+1},v_a),(v_{i+1},v_b)]$ for $0\leq i \leq k-2$ from $\Tr(G_1^{sub}, G_2^{sub})$. \\Flip the orientations of $[(v_i,v_a),(v_{i+1},v_a),(v_{i+1},v_b)]$ for all $0\leq i \leq k-2$, obtaining $[(v_i,v_a),(v_{i+1},v_b),(v_{i+1},v_a)]$. Use $\sigma_{2i}$ to denote $[(v_i,v_a),(v_i,v_b),(v_{i+1},v_b)]$, and $\sigma_{2i+1}$ to denote $[(v_i,v_a),(v_{i+1},v_b),(v_{i+1},v_a)]$.
        \item Add to the sequence another five oriented 2-simplices from $\Tr(G_1^{sub}, G_2^{sub})$ in the order as specified: $\sigma_{2k-2}=[(v_{k-1}, v_a),(v_{k-1}, v_b),(v_0, v_b)]$, $\sigma_{2k-1}=[(v_{k-1}, v_b), (v_0, v_b), (v_0, v_c)]$, $\sigma_{2k}=[(v_k, v_b), (v_0, v_b), (v_0, v_c)]$, $\sigma_{2k+1}=[(v_k, v_a), (v_k, v_b), (v_0, v_b)]$ and finally $\sigma_{2k+2}=[(v_k, v_a), (v_0, v_a), (v_0, v_b)]$.
    \end{enumerate}
    In total, we extract a sequence of $(2k+3)$ oriented 2-simplices, $T_c = \{\sigma_0, \sigma_1,\dots, \sigma_{2k+2}\}$, such that $\sigma_{i}$ and $\sigma_{i + 1\mod (2k+3)}$ share one edge. The extracted 2-simplices and their orientations as well as all shared edges are shown in Figure~\ref{fig:alngraph}(c). We flip the orientations of $[(v_i,v_a),(v_{i+1},v_a),(v_{i+1},v_b)]$ solely to ensure that the submatrix constructed below has a simple form, which makes it easier to compute the determinant.
    
    Based on $T_c$, we obtain $M_1$, a $(2k+3)\times (2k+3)$ submatrix of $[\partial]$ where each roll corresponds to a shared edge and each column corresponds to a 2-simplex in $T_c$. 
    The entry values of $M_1$ are the signed coefficients of each selected 1-simplex from the boundaries of selected 2-simplices. 
    \begin{align*}
    & M_1= \begin{bmatrix}
        1 & 0  & \dots & 0 & 0 & 0 & 0 & 0 & 0 & 1 \\
        -1 & 1  & \dots & 0 & 0 & 0 & 0 & 0 & 0 & 0\\
        0 & -1  & \dots & 0 & 0 & 0 & 0 & 0 & 0 & 0\\
        \vdots  & \vdots & \ddots & \vdots & \vdots & \vdots & \vdots & \vdots & \vdots & \vdots \\
         0 & 0  & \dots & 1 & 0 & 0 & 0 & 0 & 0 & 0\\
         0 & 0  & \dots & -1 & 1 & 0 & 0 & 0 & 0 & 0\\
         0 & 0  & \dots & 0 & -1 & 1 & 0 & 0 & 0 & 0\\
         0 & 0  & \dots & 0 & 0 & 1 & 1 & 0 & 0 & 0\\
         0 & 0  & \dots & 0 & 0 & 0 & 1 & 1 & 0 & 0\\
         0 & 0  & \dots & 0 & 0 & 0 & 0 & 1 & 1 & 0\\
         0 & 0  & \dots & 0 & 0 & 0 & 0 & 0 & -1 & -1
    \end{bmatrix}
    \end{align*}
The determinant of $M_1$ is:
\begin{align*}
    &\det M_1 \\ &= \det \begin{bmatrix}
        -1 & 1  & \dots & 0 & 0 & 0 & 0 & 0 & 0 \\
        0 & -1  & \dots & 0 & 0 & 0 & 0 & 0 & 0 \\
        \vdots  & \vdots & \ddots & \vdots & \vdots & \vdots & \vdots & \vdots & \vdots  \\
         0 & 0  & \dots & 1 & 0 & 0 & 0 & 0 & 0 \\
         0 & 0  & \dots & -1 & 1 & 0 & 0 & 0 & 0 \\
         0 & 0  & \dots & 0 & -1 & 1 & 0 & 0 & 0 \\
         0 & 0  & \dots & 0 & 0 & 1 & 1 & 0 & 0 \\
         0 & 0  & \dots & 0 & 0 & 0 & 1 & 1 & 0 \\
         0 & 0  & \dots & 0 & 0 & 0 & 0 & 1 & 1 \\
         0 & 0  & \dots & 0 & 0 & 0 & 0 & 0 & -1 
    \end{bmatrix} 
    -\det \begin{bmatrix}
        1 & 0  & \dots & 0 & 0 & 0 & 0 & 0 & 0  \\
        -1 & 1  & \dots & 0 & 0 & 0 & 0 & 0 & 0 \\
        0 & -1  & \dots & 0 & 0 & 0 & 0 & 0 & 0 \\
        \vdots  & \vdots & \ddots & \vdots & \vdots & \vdots & \vdots & \vdots & \vdots  \\
         0 & 0  & \dots & 1 & 0 & 0 & 0 & 0 & 0 \\
         0 & 0  & \dots & -1 & 1 & 0 & 0 & 0 & 0 \\
         0 & 0  & \dots & 0 & -1 & 1 & 0 & 0 & 0 \\
         0 & 0  & \dots & 0 & 0 & 1 & 1 & 0 & 0 \\
         0 & 0  & \dots & 0 & 0 & 0 & 1 & 1 & 0 \\
         0 & 0  & \dots & 0 & 0 & 0 & 0 & 1 & 1 
    \end{bmatrix} \\
    &=(-1)^{2k-2}\times (-1)-1^{2k+2}=-2.
\end{align*}
Since the determinant of $M_1$ is -2, and $M_1$ is a submatrix of $[\partial]$, $[\partial]$ is not totally unimodular.
\end{proof}

The minimal pair of input graphs that satisfy the conditions in Theorem~\ref{thm:nontu} is a graph with one 3-node cycle and one additional edge incident to the cycle and an acyclic, connected graph with three nodes. In practice, most non-trivial edge-labeled Eulerian graphs satisfy these conditions.

According to the definitions in~\citet{dey2010optimal}, the subgraph used to construct $M_1$ in the above proof (Figure~\ref{fig:alngraph}(c)) is a M\"{o}bius subcomplex, and $M_1$ is a $(2k+3)$-M\"{o}bius cycle matrix (MCM). Theorem~\ref{thm:nontu} also establishes that there may exist a M\"{o}bius subcomplex in an alignment graph, which corrects the false claim made in Lemma 2 in~\citep{gted}.


Theorem 2 in \citet{gted} attempts to employ a more algebraic approach to attempt to demonstrate that $[\partial]$ is TU by establishing that the alignment graph is a M\"{o}bius-free product space. However, the property of being M\"{o}bius-free globally does not imply the absence of M\"{o}bius subcomplexes locally. As we show in Theorem~\ref{thm:nontu}, although the alignment graph $\alngraph(G_1^{sub}, G_2^{sub})$ is homotopically equivalent to the one-dimensional circle, which is M\"{o}bius-free, it still contains a M\"{o}bius subcomplex.



\begin{figure}
    \centering
    \includegraphics[width=0.8\columnwidth]{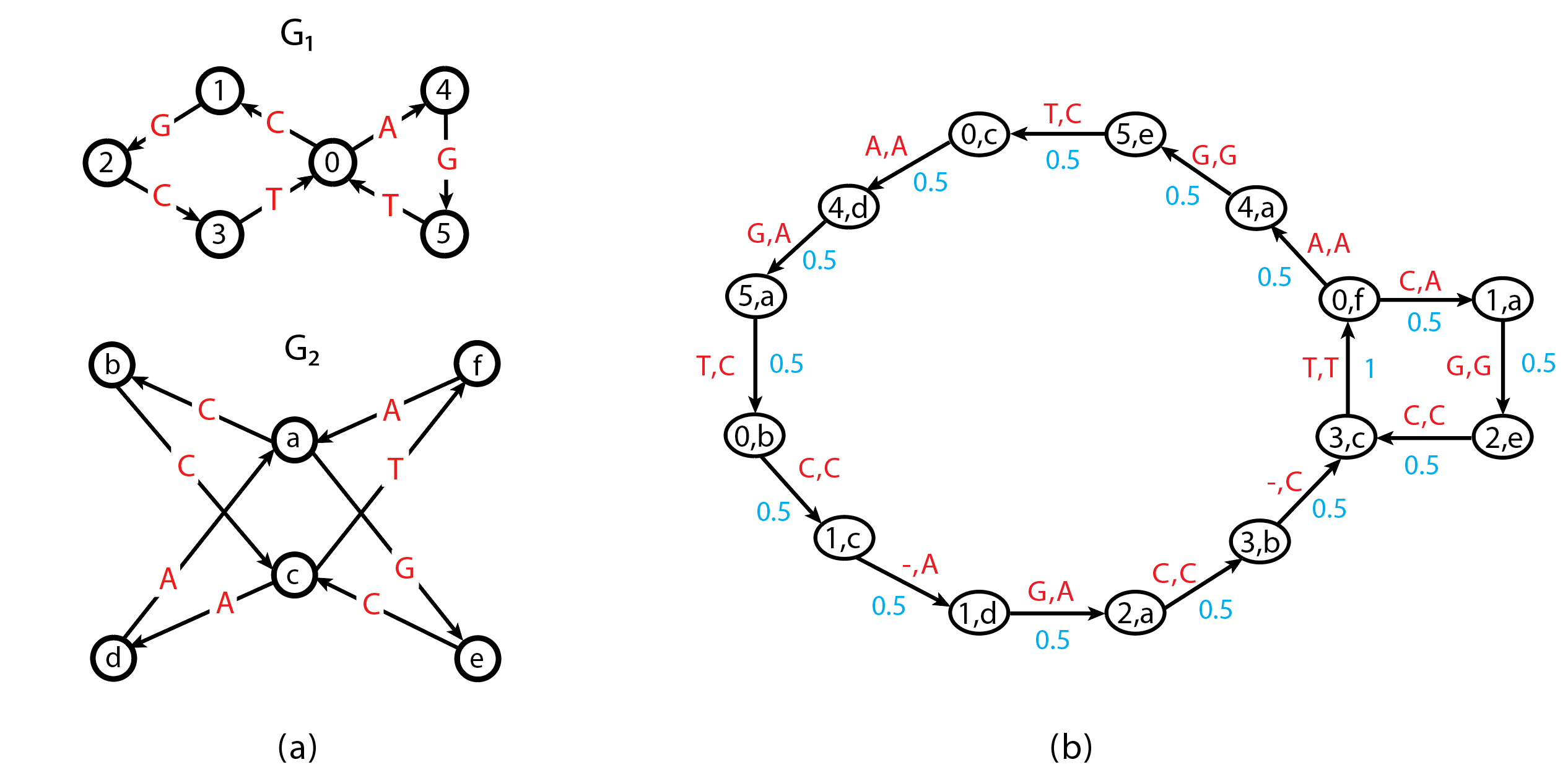}
    \caption{An example of a fractional optimal solution to \boundarylp and \gtedlp. (a) A pair of input graphs to \boundarylp and \gtedlp. Letters in red are edge labels. (b) A subgraph of $\alngraph(G_1, G_2)$ that is induced by alignment edges with non-zero weights (blue font) in an optimal solution to the LPs. The letters in red show the matching between the edge labels or between edge labels and gaps.} 
    \label{fig:ilpexp}
\end{figure}

\subsection{The LP yields optimal fractional solutions}
\label{sec:nonint}

The fact that $[\partial]$ is not totally unimodular does not guarantee that \boundarylp has a fractional optimal objective value. In this section, we prove that \boundarylp does not always yield integer optimal solutions by constructing a specific example with a fractional optimal objective value.


\begin{theorem}
    \Gtedlp and \boundarylp do not always yield optimal integer solutions.
    \label{thm:nonint}
\end{theorem}
We prove the above theorem by giving an example where \gtedlp yields a fractional optimal solution. Since by Theorem~\ref{thm:lpequal}, two LPs are equivalent, it follows that \boundarylp also yields the same fractional optimal solution.

Construct $G_1$ and $G_2$ such that their edges and edge labels are equal to the ones specified in Figure~\ref{fig:ilpexp}(a). Let the edge multi-set of $\alngraph(G_1, G_2)$ be $E$. We assign an edge cost to 0 if the edge matches two equal characters and 1 otherwise. Construct vector $x^*\in\mathrm{R}^{|E|}$ and set entries corresponding to edges in Figure~\ref{fig:ilpexp}(b) to 0.5 except edge $[(v_3,v_c),(v_0,v_f)]$ to which the corresponding entry is set to 1. Set the rest of the entries of $x^*$ to 0.

\begin{lemma}
    $x^*$ is an optimal solution to \gtedlp constructed with $\alngraph(G_1, G_2)$ and $\Tr(G_1, G_2)$.
\end{lemma}

\begin{proof}
We prove the optimality of $x^*$ via complementary slackness.
    We first write \gtedlp in standard form.
    \begin{equation}
    \begin{aligned}
        \underset{x\in \mathbb{R}^{|E|}}{\textrm{minimize}}&\quad \sum_{e\in E} \delta_e x_e\\
        \textrm{subject to}&\quad Ax = b\\
        & x_e \geq 0\quad\textrm{ for all }e\in E.
    \end{aligned}
    \label{eq:stdgtedlp}
    \end{equation}
    Here, $\delta$ is a vector of size $|E|$ where each entry is cost of edge $e$. The constraint matrix $A$ of the primal LP~\eqref{eq:stdgtedlp} has $|E|$ columns and $|V| + |E_1| + |E_2| = m$ rows, where $V$ is the vertex set of $\alngraph(G_1, G_2)$, and $E_1$ and $E_2$ are edge multi-sets of the input graphs. The first $|V|$ rows correspond to the constraints specified in~\eqref{eqn:cycle_cst}. The rest of the rows correspond to the constraints in~\eqref{eqn:proj_cst} that enforce the projected multi-set of edges to be equal to the multi-set of edges in each input graph. Since the input graphs both contain Eulerian tours, the vector $b$ has size $m$, where the first $|V|$ entries are zeroes and the rest of the entries are 1s. 
    
    We write the dual form of LP~\eqref{eq:stdgtedlp} as follows.
    \begin{equation}
        \begin{aligned}
            \underset{y\in \mathbb{R}^{m}}{\textrm{maximize}}&\quad \sum_{j=1}^{m} b_jy_j \\
            \textrm{subject to}&\quad A^{\top}y \leq \delta.
        \label{eq:dual}
        \end{aligned}
    \end{equation}
    Let the objective value of LP~\eqref{eq:stdgtedlp} given a $x$ as input is $\textrm{obj}^p_x$, and the objective value of LP~\eqref{eq:dual} given a $y$ as input is $\textrm{obj}^d_y$. To show that $x^*$ is an optimal solution to \gtedlp, we need to show that there exists a feasible solution to the dual LP, $y^*$, that satisfies the complementary slackness conditions and that $\textrm{obj}^d_{y^*} = \textrm{obj}^p_{x^{*}}$.
    
    Since each alignment edge has two endpoints and is projected to at most one edge in each graph, there are at most 4 non-zero entries in each column of $A$. The variables in $y$ of the dual form can be interpreted in three parts. Each of the first $|V|$ entries of $y$ can be assigned to each vertex in the alignment graph, and the next $|E_1|$ entries can be assigned to edges in $G_1$ and the last $|E_2|$ entries can be assigned to edges in $G_2$. There are $|E|$ constraints in the dual LP, and the $e$-th constraint can be assigned to one edge in the alignment graph has cost $\delta_e$. Therefore, each constraint that is assigned to a horizontal or a vertical edge can be written as
    \begin{equation}
        y_{v^{out}_e} - y_{v^{in}_e} + y_{e_i} \leq \delta_e,
    \end{equation}
    where $i=1$ if $e$ is a horizontal edge, and $i=2$ if $e$ is a vertical edge. $y_{v^{in}_e}$ and $y_{v^{out}_e}$ are the $y$ entries that are assigned to the vertices that are the start and end of edge $e$, and $y_{e_i}$ are the $y$ entries that assigned to the $\pi_i(e)$.
    
    Similarly, each constraint that is assigned to a diagonal edge is
    \begin{equation}
        y_{v^{in}_e} - y_{v^{out}_e} +y_{e_1} + y_{e_2} \leq \delta_e.
    \end{equation}
    
    We can verify that $x^*$ is a feasible solution of the primal form~\eqref{eq:stdgtedlp} by checking if constraints~\eqref{eqn:cycle_cst}-\eqref{eqn:proj_cst} are satisfied. The primal objective value can be computed in a straightforward way, and we can obtain $\textrm{obj}^p_{x^*} = 3.5.$

    According to complementary slackness conditions, since $x^*_e > 0$ for edges shown in Figure~\ref{fig:ilpexp}(b), the corresponding constraints in the dual LP~\eqref{eq:dual} must be tight, meaning that the equality must hold in these constraints. The rest of the dual constraints could have slacks.
    
    Let the subgraph of $\alngraph(G_1, G_2)$ shown in Figure~\ref{fig:ilpexp}(b) be $A'$.
    Denote the cycle that traverses from $[(0,f), (4,a)]$ to $[(3,c),(0,f)]$ be $C'$ and the 4-node cycle that traverses $((0,f), (1,a),(2,e), (3,c))$ be $C''$. Denote the concatenation of two cycles with $C$. The projected cycle from $C$ to $G_1$ is
    \begin{align}
        C_1 = (v_0, v_4, v_5, v_0, v_4, v_5, v_0, v_1, v_2, v_3, v_0, v_1, v_2, v_3, v_0).
    \end{align}
    The projected cycle from $C$ to $G_2$ is
    \begin{align}
        C_2 = (v_f, v_a, v_e, v_c, v_d, v_a, v_b, v_c, v_d, v_a, v_b,v_c, v_f, v_a, v_e, v_c, v_f).
    \end{align}
    Sum up all the constraints that are assigned edge $e$ where $x^*_e > 0$. Since these edges form a cycle, we get:
    \begin{align}
        & \sum_{e\in C}\big(y_{v_e^{out}} - y_{v_e^{in}}\big) + 2\big(\sum_{e_1\in C_1} y_{e_1} + \sum_{e_2\in C_2} y_{e_2}\big)\\
        = \quad & 0 + 2\big(\sum_{e_1\in C_1} y_{e_1} + \sum_{e_2\in C_2} y_{e_2}\big)\\
        = \quad &\sum_{e\in C} \delta_e  = 7,\\
        \Rightarrow \quad& \sum_{e_1\in C_1} y_{e_1} + \sum_{e_2\in C_2} y_{e_2} = 3.5.
        \label{eq:finaldualcst}
    \end{align}
    The summed edge cost is 7 as there are 7 edges that are either mismatch edges or vertical edges.

    All $y$ entries that correspond to vertices are free variables and are in every constraint. After fixing the $y$ variables that satisfy constraint~\eqref{eq:finaldualcst}, the rest of the $y$ variables can be set to satisfy the dual cosntraint. We now obtain $y^*$ which is a feasible solution to the dual LP.

    The only entries in $y^*$ that could have non-zero dual costs are those that correspond to edges in $E_1$ and $E_2$. Since these corresponding dual costs are all 1, 
    
    \begin{equation*}\begin{aligned}[t]
        \textrm{obj}^d_{y^*} = \sum_{e_1\in C_1} y_{e_1} + \sum_{e_2\in C_2} y_{e_2} = 3.5 = \textrm{obj}^p_{x^*}.
    \end{aligned}
    \end{equation*}
\end{proof}

Since the costs of alignment graph edges are all integers, the fact that \boundarylp and \gtedlp yield fractional optimal objective values mean that they must yield fractional solutions and assign fractional values to entries in $x$. Theorem~\ref{thm:nonint} follows. Since \boundarylp yields fractional solutions and GTED is always an integer, solving \boundarylp does not solve GTED.

\section{The average wall-clock time to solve ILPs on 3-cycle graphs}

\begin{table}[!h]
\renewcommand\thetable{S1}
\centering
\caption{The average wall-clock time to solve~\eqref{eqn:lowboundilp},~\eqref{eqn:newgtedlp_obj},~\eqref{eqn:compactgtedilp_obj} and the number of iterations for pairs of 3-cycle graphs for each $\GTED-\GTED_l$.}
\label{tab:runtime}
\begin{tabular}{@{}ccccc@{}}
\toprule
\textbf{GTED - $\GTED_l$} & \textbf{\begin{tabular}[c]{@{}c@{}}\eqref{eqn:lowboundilp} \\ runtime (s)\end{tabular}} & \textbf{\begin{tabular}[c]{@{}c@{}}GTED iterative \\ runtime (s)\end{tabular}} & \textbf{Iterations} & \textbf{\begin{tabular}[c]{@{}c@{}}GTED compact\\ runtime (s)\end{tabular}} \\ \midrule
\textbf{1.0}          & 0.06                                                                  & 0.17                                                                           & 3.55                & 0.39                                                                        \\
\textbf{2.0}          & 0.05                                                                  & 0.87                                                                           & 13.00               & 0.43                                                                        \\
\textbf{3.0}          & 0.08                                                                  & 25.41                                                                          & 67.60               & 1.24                                                                        \\
\textbf{4.0}          & 0.07                                                                  & 205.59                                                                         & 179.10              & 1.70                                                                        \\
\textbf{5.0}          & 0.08                                                                  & 943.68                                                                         & 502.85              & 5.37                                                                        \\ \bottomrule
\end{tabular}
\end{table}

\end{appendices}

\end{document}